\documentclass[12pt]{article}
\usepackage{graphicx}
\usepackage{subcaption}
\usepackage{amsmath,amssymb,setspace,amsfonts,enumitem,fullpage,amsthm,comment,color,bm}
\usepackage[longnamesfirst]{natbib}
\usepackage {palatino}
\usepackage[super]{nth}
\usepackage{epsf}
\usepackage[margin=10pt,labelfont=bf,labelsep=endash]{caption}
\usepackage[colorlinks,pagebackref=false]{hyperref}
\usepackage[usenames,dvipsnames]{xcolor}
\usepackage{titlesec}
\usepackage{dsfont}
\usepackage{chngcntr}
\usepackage{apptools}
\usepackage{epstopdf}
\AtAppendix{\counterwithin{lemma}{section}}
\usepackage{footmisc}
\AtAppendix{\counterwithin{lemma}{section}}
\AtAppendix{\counterwithin{proposition}{section}}
\graphicspath{{./Graphs/}}

\allowdisplaybreaks 

\definecolor{dark-red}{rgb}{0.4,0.15,0.15}
\definecolor{dark-blue}{rgb}{0.15,0.15,0.75}
\definecolor{medium-blue}{rgb}{0,0,0.5}
\hypersetup{
    pdftitle={Pandora's Ballot Box},
    pdfauthor={Buisseret, Van Weelden},
    citecolor=dark-blue,
    bookmarksnumbered=true,
    urlcolor=BrickRed,
    linkcolor=dark-blue
}

\makeatletter
  \renewcommand\@seccntformat[1]{\csname the#1\endcsname.{\hskip.7em\relax}} 
\makeatother
\makeatletter
\renewenvironment{proof}[1][\proofname] {\par\pushQED{\qed}\normalfont\topsep6\p@\@plus6\p@\relax\trivlist\item[\hskip\labelsep\bfseries#1\@addpunct{.}]\ignorespaces}{\popQED\endtrivlist\@endpefalse}
\makeatother

\newtheorem{corollary}{Corollary}

\newtheorem{proposition}{Proposition}

\newtheorem{remark}{Remark}

\theoremstyle{definition}

\titlespacing\section{0pt}{10pt plus 2pt minus 2pt}{4pt plus 2pt minus 2pt} 
\titlespacing\subsection{0pt}{6pt plus 2pt minus 2pt}{2pt plus 2pt minus 2pt} 
\titlespacing\subsubsection{0pt}{6pt plus 2pt minus 2pt}{0pt plus 2pt minus 2pt} 

\renewcommand{\epsilon}{\varepsilon}

\let\oldfootnote\footnote
\renewcommand\footnote[1]{\oldfootnote{\hspace{.5mm}#1}}

\setcitestyle{aysep={}}

\usepackage{setspace}

\renewcommand{\bar}{\overline}

\begin{document}

\begin{titlepage}

\title{Pandora's Ballot Box: Electoral Politics of Direct Democracy\footnote{We are grateful to Daniel Carpenter, Francisco Espinosa, Amanda Friedenberg, Ryan Hubert, Gilat Levy, Luis Martinez, and seminar audiences at Harvard, Georgetown, Yale, NYU, Arizona, Federal University of Rio Grande do Sul, the IE Political Economy Annual Workshop, PolEconUK,  EUI, Barcelona GSE, and UQAM.}}   
\author{ \;\;\;\;Peter Buisseret\footnote{Department of Government, Harvard University, \emph{Email}: \url{ pbuisseret@fas.harvard.edu}} \,\,\,\,\,Richard Van Weelden\footnote{Department of Economics, University of Pittsburgh, \emph{Email}: \url{rmv22@pitt.edu}}}

\maketitle

\begin{abstract}
We study how office-seeking parties use direct democracy to shape elections. A party with a strong electoral base can benefit from using a binding referendum to resolve issues that divide its core supporters. When referendums do not bind, however, an electorally disadvantaged party may initiate a referendum to elevate new issues in order to divide the supporters of its stronger opponent. We identify conditions under which direct democracy improves congruence between policy outcomes and voter preferences, but also show that it can lead to greater misalignment both on issues subject to direct democracy and those that are not.
\end{abstract}

\begin{center}
\textbf{Word Count}: 10,500
\end{center}

\thispagestyle{empty}
\end{titlepage}



\let \markeverypar \everypar
\newtoks \everypar
\everypar \markeverypar
\markeverypar{\the \everypar \looseness=-2\relax}
\newpage
\setcounter{page}{1}
\begin{quote}{\emph{``This is a vital decision for the future of our country and I believe we should also be clear that it is a \emph{final} decision.''}}\end{quote}
\vspace{-10mm}
\begin{flushright}
David Cameron, 2016
\end{flushright}
\vspace{-8mm}
\begin{quote}{\emph{``\emph{Finality}, sir, is not the language of politics.''}}\end{quote}
\vspace{-10mm}
\begin{flushright}
Benjamin Disraeli, 1859
\end{flushright}
\vspace{-8mm}

\section*{Introduction}

Direct democracy occupies a central role in contemporary policymaking around the world. Under representative democracy, citizens have only one vote to cast for politicians who choose policies on many issues. By \emph{unbundling} these issues \citep{besley2008issue}, direct democracy can strengthen congruence between policies and voters' preferences \citep{matsusaka2005direct}. Existing work identifies a range of mechanisms by which  direct democracy helps voters to sanction, circumvent, and discipline politicians (e.g., \citealp{gerber1996legislative}, \citealp{matsusaka2001political}, \citealp{MT:04}, \citealp{prato2017hidden}, \citealp{le2018popular}). 

While voters can use direct democracy to constrain policymakers, it is often the politicians themselves who resort to direct democracy. Of the 57 referendums on European integration held across European Union (EU) member states, 24 were initiated in the absence of a constitutional mandate \citep{oppermann2017derailing}. Recent examples include the United Kingdom's (UK) Brexit referendum of 2016, Hungary's 2016 referendum on EU migrant quotas, and the 2015 Greek Bailout referendum. Governments held discretionary referendums on international trade agreements (Costa Rica 2007, Netherlands 2016), social policies (Australia 2017) and territorial governance (Canada 1980 and 1995, Scotland 2014), with at least one political party pushing for the referendum. Politicians even shape and promote ballot initiatives. In 1994, the California Republican Party's support for Proposition 187\footnote{Proposition 187, passed by voters but subsequently ruled unconstitutional, prohibited undocumented immigrants from accessing health care, public education, and other government services in the state of California.} was critical in getting the measure on the ballot: it provided almost half of total contributions \citep{hasen2000parties} and also ``rolled out a 200,000 piece fundraising mailing for the initiative and helped field an army of paid signature gatherers to finish the task of qualification'' \citep[88]{Schultz1996}.


A variety of strategic motives may drive politicians' pursuit of direct democracy. In some cases, a referendum is a legal requirement to implement a particular policy.\footnote{Countries with provisions for legally binding referendums include France, Georgia, Greece, Ireland, Russia and Switzerland.  In the United States certain policies, from constitutional amendments to raising property taxes and issuing school bonds, often require direct voter approval.}   It can be used to strengthen a government's hand in international negotiations \citep{putnam1988diplomacy, schneider1994change}, or resolve disputes between different government branches \citep{breuer2008problematic, tridimas2007ratification} or between different parties or factions \citep{hug2004occurrence, rahat2009elite, morel2001rise, qvortrup2006democracy}.  
We focus on domestic \emph{electoral} considerations of \emph{discretionary} referendums. Our paper asks: under what circumstances can an office-motivated party use direct democracy to improve its electoral prospects? How does this depend on the initial degree of polarization both \emph{between} and \emph{within} parties? 
And, how does it depend on a party's initial electoral advantage or disadvantage?  

The prevailing electorally-motivated account of referendums is that politicians use them as a ``lightning rod'' \citep{bjorklund1982demand} for conflict on issues that cut across traditional party lines. This literature views referendums as a way to take difficult issues out of the electoral arena, since binding policymakers to whichever outcome wins a majority  ``removes the policy from the realm of discretionary government choices'' \citep[246]{oppermann2017derailing}. This recourse is especially valuable for a relatively divided party.  The reason is that by isolating the conflict away from electoral politics, voters  ``need not have their usual party allegiances distorted'' \citep[448]{aylott2002let}.  

We contend that this account is critically incomplete because direct democracy frequently does \emph{not} bind policymakers.  Even if legally binding, policymakers can often circumvent democratic mandates \citep{gerber2004direct}.  For example, in California's 2010 Attorney General Election, Kamala Harris and Steve Cooley divided on whether to defend Proposition 8, which banned same-sex marriage, in any future lawsuit.\footnote{See ``Kamala Harris-Steve Cooley race could affect Prop. 8''. \emph{SFGate}, November 8, 2010. \url{https://rb.gy/fmgqmd}.} Harris won, and Proposition 8 was invalidated when California chose not to defend the law in court. Proposition 187, mentioned above, met a similar fate after Gray Davis became governor in 1999 and withdrew California's legal appeals.  Florida's Amendment 4 called for restoring felons' voting rights, but subsequent legislation by the Republican-led Senate  included an exemption for felons with unpaid fines, greatly reducing the effect of the  Amendment.  After a majority of voters in Colombia's 2016 referendum rejected the government's peace agreement with FARC, the government implemented a revised peace agreement legislatively. Similarly, the European Union implemented a revised European Constitution via 
the Treaty of Lisbon after the initial draft was rejected by French and Dutch voters. 

When direct democratic mandates do not bind policymakers, direct democracy \emph{can} settle conflicts but often does \emph{not}. A clear example is the central role of the UK's 2016 Brexit referendum in British politics for several years and the next two General Elections.  In the referendum's aftermath, the major parties subsequently adopted diverging positions, with the Conservatives viewed as more able to `get Brexit done'.   Given that 65 per cent of Labour supporters voted to \emph{Remain},\footnote{\emph{How Britain voted at the EU referendum}, YouGov, June 27 2016, \url{https://rb.gy/vh4dba}.} Labour's Deputy Leader Tom Watson concluded in the 2019 election that ``the simple truth is---whatever anyone says---that Labour is a `Remain' party.''\footnote{See {\emph{Jeremy Corbyn faces calls to resolve Labour Brexit divisions}, BBC, 22 September 2019.}} Similarly, referendums in Quebec (1980) and Scotland (2014) did not settle independence; in 1995 another Quebec referendum was held, and, on June 28, 2022, backed by a majority of the Scottish Parliament, First Minister Nicola Sturgeon announced plans for a second independence referendum in 2023. 

Finally, direct democracy is often not intended to settle conflicts---on the contrary, it is frequently wielded to \emph{intensify} their electoral salience.  Hungary's 2016 EU migrant settlement referendum reflected incumbent Fidesz's goal of winning back support from radical right alternative Jobbik (\citealp{gessler20172016}).  In 1996 Californian Republicans promoted Proposition 209, banning affirmative action, in order to ``split Democratic support for President Bill Clinton'' \citep{smith2001initiative}. In Switzerland, ``[t]he populist right-wing party ... relied extensively on questions pertaining to the relationship with the EU and immigration---on both those topics their immediate competitors...held positions that were not fully backed by their supporters'' \citep[613]{leemann2015political}. And, it is likely that Taiwan's 2004 referendum on relations with China was timed to coincide with nationalist incumbent Chen Shui-bian's presidential re-election contest \citep{rawnsley2005peaceful}. 

%

Motivated by these contentions, we develop a new theoretical framework to analyze the electoral aspect of direct democracy. Our model features two dimensions of policy: a \emph{traditional} or `partisan' issue that defines the parties, and a \emph{second} or `emerging' issue that divides voters both within and across parties. The traditional issue is most naturally interpreted as a partisan issue conflict such as redistribution, or liberalism versus conservatism. The second issue could capture disagreement over an international treaty, same-sex marriage, or separatism which is subject to direct democracy.

Our main presentation considers competition between two parties, subsequently extending our analysis to incorporate a third party. Each voter aligns with one of the two parties---`Left' or `Right'---on the traditional issue, and so she is one of that party's `core supporters'.  Voter preferences are heterogeneous on the second issue, even among each party's core supporters, and there is aggregate uncertainty about both the direction and intensity of these preferences.

Parties always implement their core supporters' majority-preferred traditional policy. Absent direct democracy, parties also implement what they believe to be their supporters' majority-preferred policy on the second issue. We consider two perspectives on how policy is made after a referendum or initiative: a \emph{binding} context in which parties are forced to implement whichever alternative secures a majority, and a \emph{non-binding} context in which parties continue to implement their supporters' majority-preferred outcome even after a referendum. Thus, unlike a binding referendum which commits parties to the winning policy, a non-binding referendum simply reveals information about voter preferences on the emerging issue.  This information may, in turn, alter a party's positioning on that issue.  While we view the non-binding feature as central in most referendums, binding referendums are the bedrock of existing theoretical work on direct democracy and so are an important benchmark.  

Our framework focuses on a single decision: whether or not to call a referendum on the second issue. We evaluate this choice from the perspective of a leader whose objective is to secure the election of one of the two parties, which we take to be the Right party. This could be the party leader or national committee, or it could be an influential donor deciding whether to commit resources to support a ballot initiative.  

We start with a binding benchmark in which policymakers are forced to implement the majority-preferred policy after any referendum, reflecting either constitutional fiat or un-modeled political constraints. While the parties' core supporters always diverge on the partisan issue, their majorities may initially either align or mis-align on the second issue, which determines whether the second issue is initially salient in the election.

If the parties' majorities initially divide on the second issue, absent a referendum the election is fought on both the traditional \emph{and} second issue. A fraction of each party's core supporters mis-align with their traditional party on the second issue, and some may feel strongly enough to vote for the other party.  A party in the majority---with more voters to lose---or a party that is evenly divided on the second issue---and so likely to lose a larger share of its base---is most vulnerable in a multi-issue contest. 

When direct democracy mandates bind, both parties adopt the majority-preferred policy on the second issue after a referendum, removing it from the election and focusing the contest solely on the partisan issue. We therefore predict that binding referendums are most likely to be initiated by parties with a majority of support on the partisan issue, and that are more divided on the emerging issue. This is consistent with the traditional account of referendums as `lightning rods'.  However, we derive new predictions about political contexts in which they are used.  For example, because higher polarization reduces cross-over voting in a multi-issue election, it reduces a more divided party's incentive  to initiate a referendum with the goal of taking the second issue off the table.

We then transition to a non-binding setting, in which we assume that parties implement their core supporters' majority-preferred policy with or without a referendum.  As a consequence, not only may referendums fail to restore single-issue elections by settling inter-party conflicts on the second issue, but they may also \emph{create} multi-issue elections by revealing unanticipated conflicts between the parties' core supporters.

We continue to predict that divided parties  favor a referendum to resolve the second issue when the parties initially misalign. Surprisingly, given that the referendum may not succeed in resolving the second issue, this incentive can be even \emph{stronger} for the relatively divided party than with a binding  referendum that always settles the second issue.  
To see why, recognize that even a minority party wants to resolve the second issue if it is sufficiently divided. A non-binding referendum fails to resolve the issue only when voters in both parties are evenly divided on the second issue. This leads a larger share of voters to cross party lines in the subsequent election, to the net benefit of the minority party, increasing  its willingness to gamble on a referendum. 
%


We also identify a novel context in which referendums may be electorally valuable: when the parties' majorities initially align on the second issue, but one party is at a partisan disadvantage.  This is because a referendum may transform a single-issue election into a multi-issue contest if it mis-aligns the parties' majorities, and a sufficiently disadvantaged party gambles on a referendum in the hopes of dividing its stronger opponent. This gambling incentive is strongest for a weaker party that is relatively united on the second issue, and particularly when polarization between the parties on the traditional issue is relatively low. This is because low inter-party polarization makes the stronger party's core supporters more susceptible to the minority party's appeal on the secondary issue.

We then assess the normative appeal of direct democracy. A prominent literature \citep{besley2008issue, matsusaka2005direct} emphasizes how direct democracy enhances voter welfare by unbundling policies and forcing policymakers to take the majority-preferred position.  We show that this is not always the case, particularly when the referendum is non-binding.  When the parties' majorities initially mis-align on the second issue, a referendum always increases policy alignment with the majority on the second issue, but it may reduce alignment on the traditional issue.  When the parties initially align on the second issue, a non-binding referendum may reduce the probability of the majority-preferred policy on \emph{both} issues.  While previous work highlights the risks of special interest influence \citep{broder2000democracy, matsusaka2001political}, low turnout \citep{RSS:22}, or an ill-informed electorate \citep{gerber1995campaign, boehmke2007selection}, we show that direct democracy can lead to policy misalignment \emph{even if} voters always choose the majority-preferred policy in the referendum.  


Our emphasis on non-binding referendums distinguishes our analysis from existing theoretical accounts, which view binding mandates as the \emph{defining} feature of direct democracy \citep{besley2008issue, MT:04}.  
The few papers on how politicians can circumvent direct democratic outcomes at a subsequent legislative or implementation phase \citep{gerber2004direct}, and how voters can learn from a policymaker's choice of whether to do so \citep{xefteris2011referenda}, abstract from the question of which referendums will be proposed. This question is the focus of our analysis.

We close by outlining some extensions. We 
introduce a third party that emulates one of the two parties on the traditional issue and show when a party may favor a referendum to diffuse  the third-party threat. We also consider direct democracy as an instrument to mobilize partisan turnout---a widely-attributed strategic motive in the United States, in which initiatives typically appear on the ballot simultaneously with general elections.  

\section*{Model} 

\label{s:model}


\noindent\textbf{Preliminaries.} Two parties, Left ($L$) and Right ($R$), compete for the support of a unit mass of voters. There are two dimensions of policy disagreement, $x$ and $y$, and in each dimension there are two policy positions, $0$ and $1$. The traditional issue dimension is most naturally interpreted as a left-right conflict, such as high levels of redistribution ($x=0$) versus low redistribution ($x=1$).\footnote{For other papers in which parties are fixed at a given policy on a traditional issue see \citet{KP:14} and \citet{BVW:20, BVW:21}.}  We interpret the second dimension as an emerging issue cleavage. It could capture participation in an international agreement or national union;  it could alternatively reflect a social cleavage, such as abortion or same-sex marriage. \autoref{fig:a} illustrates the policy dimensions.

\begin{figure}[t!]
\centering
  \includegraphics[width=.45\linewidth]{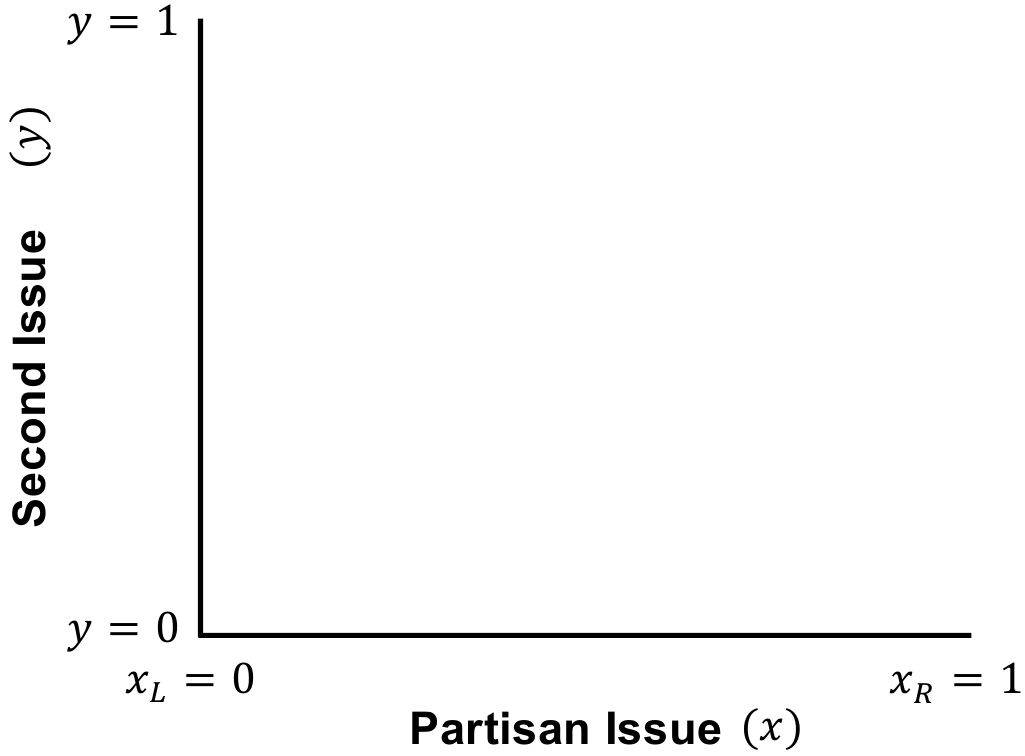}
\caption{Two Dimensions of Policy.} 
\label{fig:a}
\end{figure}

\noindent \textbf{Voter Preferences.} A fraction $\mu$ of voters are \emph{policy}-oriented voters, and a fraction $1-\mu$ are \emph{noise} voters (\citeauthor{besley2008issue}, \citeyear{besley2008issue}). The role of noise voters is solely to introduce tractable aggregate uncertainty in election outcomes. We describe the preferences and behavior of each group in turn.

\noindent\emph{Policy Voters.} Each policy-oriented voter $i$ derives net payoff $p>0$ from her preferred policy $x^i\in\{0,1\}$ on the traditional issue. She also derives a net benefit ${b}^i\in\mathbb{R}$ from policy $y=1$. Thus, if policy $(x,y)\in\{0,1\}^2$ is implemented, policy voter $i$'s payoff is 
\[
u(x,y;x^i,b^i)=p\times \begin{cases}
1 &\text{if $x=x^i$}\\
0 &\text{if $x\neq x^i$}
\end{cases}
\;\;+ {b}^i \times y.
\]
A fraction $r\in(0,1)$ of voters are `core supporters' of the Right party: they are \emph{conservatives} with ideal policy $x^i=1$.  The remaining fraction, $1-r$, of voters are \emph{liberals} with ideal policy $x^i=0$---the core supporters of the Left party. For simplicity, we assume $r$ is known.

Policy preferences are heterogenous within each party on the $y$ dimension.  Among the core supporters of party $J\in\{L,R\}$, individual preference ${b}^i$ is drawn from a symmetric distribution with mean $b_J+\gamma$.  So $b_L+\gamma$ is the average (and median) liberal's value from $y=1$ and $b_R+\gamma$ is the average/median value for conservatives.  Parameters $b_L$ and $b_R$ capture differences in attitudes on the emerging issue between the two parties' core supporters. The random variable $\gamma$ is a common stochastic shock, capturing changes in attitudes and preferences on the emerging issue. It is drawn from a strictly increasing mean-zero distribution $G(\cdot)$ with symmetric, continuously differentiable, and strictly quasi-concave density $g$. As $\gamma$ is initially unknown, so too is the support for $y=1$ in each party.  A large positive (negative) $\gamma$ realization implies that average attitudes in \emph{both} parties shift in favor of $y=1$ (respectively $y=0$).

We assume $b^i-{b}_J-\gamma$ is drawn from a strictly increasing mean-zero and logconcave distribution $B(\cdot)$, with symmetric and continuously differentiable density $b(\cdot)$. 
To economize on parameters, we assume that $B(\cdot)$ and $G(\cdot)$ have full support on $\mathbb{R}$.

\noindent\emph{Noise voters.}  A fraction $1-\mu$ of voters are noise voters.  Of these, a uniform random share $\eta\in[0,1]$ vote for Right in the general election, and the remainder for Left, no matter the policies of each party.   For simplicity we assume that noise voters do not participate in a referendum, if one is held.

\noindent\textbf{Party Objectives and Platforms.} Throughout, we assume that the party leadership seeks to maximize the probability of winning the election.\footnote{This could reflect purely office-seeking motives. However, it is equivalent to assuming the leadership maximizes the probability that the party's traditional policy ($x=0$ for Left, $x=1$ for Right) wins the election. Thus, it could reflect the priorities of influential party leaders or donors who overwhelmingly prioritize securing their preferred traditional policy (e.g. low taxes).}

We assume that each party implements the policy expected to be preferred by a majority of its core supporters on any issue not subject to referendum.\footnote{This could be micro-founded as follows: prior to the election, a primary is held in each party to determine the party's leader. In that case, office-seeking candidates converge to the most-preferred policy of the (expected) median party supporter.}   Thus, on the traditional issue, Left always implements $x=0$ and Right always implements $x=1$. Similarly, on the emerging issue, party $J\in\{L,R\}$ implements $y=1$ if and only if $b_J \geq 0$.\footnote{If the expected median benefit is $0$ the party's median is equally likely to prefer either position, and we assume the party breaks the tie in favor of $y=1$.}

We adopt different perspectives on how parties make policy on the second issue after a referendum.  If the referendum binds, both parties are constrained to propose the winning policy.  If the referendum does not bind, however, party $J$ proposes $y=1$ if and only if the expectation of $b_J+\gamma$, given the information revealed in the referendum, is positive.  

\noindent \textbf{Timing.} At the outset, the incumbent party (Right) leadership chooses whether to hold a referendum on the second issue. If a referendum is held, all policy voters cast a ballot for either $y=1$ or $y=0$, and the outcome is publicly observed. An election is then held, in which each voter casts her ballot for either Left or Right. The majority-winning party implements its policy $(x,y)\in\{0,1\}^2$. To summarize:

\begin{itemize}\vspace{-2mm}
\item[0.] Nature draws random variables $\gamma$ and $\eta$. Neither of these realizations are observed by any agent.
\item[1.] \emph{Right}'s leadership decides whether to hold a referendum on the second issue.
\item[2.] Parties commit to policies on each issue.  
\item[3.] A general election is held, and the majority-winning party implements policy.
\end{itemize}

\noindent \textbf{Voter Behavior.} 
We assume that policy voters cast their ballots \emph{sincerely} in both the referendum and the general election. In other words: a policy voter always votes for the policy or candidate that gives her a higher payoff.  In the referendum this means voters do not consider how their individual vote in the referendum may impact the general election, and each voter $i$ supports $y=1$ if and only if $b^i \geq 0$.  In the general election sincere voting only rules out the weakly dominated strategy of voting for the less preferred candidate.

\noindent \textbf{Parametric Assumptions.} Finally, we make two parametric assumptions to focus our analysis on the most interesting cases.

\noindent \emph{Assumption 1.} $b_L<0$ and $b_L < b_R$.

Our first assumption states that a majority of Left members initially favor policy $y=0$ on the second issue, and a larger fraction of voters in the Right party favor $y=1$ than in Left.   Assumption 1 imposes no loss of generality beyond $b_L \neq b_R$, which ensures that for some realizations of $\gamma$ the parties disagree on the $y$ policy.  It is without loss of generality because the labeling of the second dimension is arbitrary, and the game is zero-sum between the two office-seeking parties, so Left  benefits from a referendum if and only if Right does not.  

\noindent \emph{Assumption 2.} $1-\frac{1}{2\mu}<r<\frac{1}{2\mu}$.

Assumption 2 allows for $r \neq 1/2$, but requires that the parties' shares of core supporters cannot be too imbalanced.  It ensures the share of noise voters is large enough that election outcomes are uncertain.

\noindent \textbf{Equilibrium.} 
Given sincere voting, and that party positions are pinned down by the preferences of their core supporters, the only strategic choice is Right's decision of whether to call a referendum.  Our analysis characterizes that optimal decision. 

\section*{Benchmark: Binding Referendums}

We initially maintain the assumption from previous theoretical analyses of direct democracy \citep{besley2008issue, MT:04} that the government is compelled to implement whichever policy wins a majority in a referendum. 


Recall that  a voter $i$ favors $y=1$ if and only if $b^i\ge0$. Since voters' net values for $y=1$ in party $J\in\{L,R\}$ are symmetrically distributed around $b_J+\gamma$, the fraction of voters in party $J$ who prefer $y=1$ is $1-B(-\gamma-b_J)=B(\gamma+b_J)$.   Thus, a majority of voters across the two parties favor $y=1$ if and only if
\begin{equation}
r B(\gamma+b_R)+(1-r) B(\gamma+b_J)\ge 1/2.\label{gammas}
\end{equation}
As the left hand side of \eqref{gammas} is continuous and strictly increases in $\gamma$ there is a unique threshold $\gamma^*(r,b_R,b_L)\in (-b_R,-b_L)$ such that a majority of voters favor $y=1$ if and only if $\gamma\ge \gamma^*$. So, if a referendum is held, alternative $y=1$ wins if and only if $\gamma$ exceeds $\gamma^*$.  

If a binding referendum is held, whichever party wins the election implements the majority-winning policy, and so parties converge on the emerging issue.  Of course, if $b_L < b_R<0$, a majority in both parties are expected to prefer $y=0$ initially, and so the parties converge absent a referendum as well.  
In either case, the subsequent election divides the parties solely on the traditional (i.e., liberal-conservative) cleavage. 

Among the share $\mu$ of policy voters, a fraction $r\in(0,1)$ are conservatives who prefer Right, and the remaining $1-r$ are liberals who prefer Left. Recalling that Right receives a uniform random share $\eta$ of the remaining $1-\mu$ noise voters, Right wins if and only if $\mu r+(1-\mu)\eta$ exceeds $\mu(1-r)+(1-\mu)(1-\eta)$, which occurs with probability

\begin{equation}
\lambda(r)\equiv \frac{1}{2}+\frac{\mu}{1-\mu}\left[r-\frac{1}{2}\right].\label{lambda}
\end{equation}
If $b_R<0$ then Right wins with probability $\lambda(r)$ whether or not a binding referendum is held.  

There is one case, however, in which a binding referendum changes election probabilities. When $b_R\ge 0>b_L$, a majority of liberals are anticipated to favor $y=0$, and a majority of conservatives to favor $y=1$. Absent a referendum, the parties' policies differ on both issues, leading to a multi-issue election. In that election, a conservative voter $i$ supports Right if and only if her combined net benefit $b^i$ from $y=1$ and a conservative traditional policy ($p$) is positive, i.e., she votes for Right if and only if $p+b^i \ge 0$. Similarly, a liberal voter $i$ supports Left if and only if she does not mis-align too strongly with Left's policy of $y=0$, i.e., if and only if $p \geq b^i$. The share of conservatives who support Right is thus $1-B(-p-\gamma-b_R)=B(p+\gamma+b_R)$ and the share of liberals who support Right is $1-B(p-\gamma-b_L)=B(\gamma+b_L-p)$.  So, Right's share of policy voters is: 
\begin{equation}
s(\gamma;r,b_R,b_L,p)\equiv r B(\gamma+b_R+p)+(1-r) B(\gamma+b_L-p).\label{eq:noref}
\end{equation}
As this share of policy voters depends on the realization of $\gamma$, Right's probability of winning the election in the absence of a referendum is 
\begin{equation}
\int_{-\infty}^{\infty}\lambda(s(\gamma;r,b_R,b_L,p))g(\gamma)\,d\gamma, \label{eq:multi}
\end{equation}
where the function $\lambda(\cdot)$ is defined in expression \eqref{lambda} and $s(\cdot)$ is defined in \eqref{eq:noref}.

The effect of a binding referendum, which turns a multi-issue election into a single-issue one, on Right's probability of winning is then  
\begin{align}
D(r, b_L, b_R, p)&=\int_{-\infty}^{\infty}[\lambda(r)-\lambda(s(\gamma;r,b_R,b_L,p))]g(\gamma)d\gamma  \propto\int_{-\infty}^{\infty}(r-s(\gamma, r, b_L, b_R, p))g(\gamma)d\gamma. \label{D}
\end{align}
Right wants to initiate a referendum if and only if $D(r, b_L, b_R, p)>0$.  The next proposition characterizes when this is satisfied.

\begin{proposition}\label{pro:bench}
If the parties initially \textbf{align} on the second issue ($b_R<0$) a binding referendum has no effect on Right's election prospects. If the parties \textbf{misalign} on the second issue ($b_R\ge 0$), there exists a threshold $r_{\text{bind}}\in(0,1)$ such that Right strictly benefits from a binding referendum if and only if ${r}>r_{\text{bind}}$. Threshold $r_{\text{bind}}$:

\noindent(1) \emph{strictly increases} in $b_R$ and $b_L$ and equals $1/2$ when $b_R=-b_L$.


\noindent(2) \emph{strictly increases} in $p$ if $b_R > -b_L$ and \emph{strictly decreases} in $p$ if $b_R < -b_L$.

\end{proposition} 

\autoref{pro:bench} shows that, when the parties initially divide on the second issue, a binding referendum increases a party's prospects if its partisan base is sufficiently strong.  Because all voters support their own party in a single-issue election, but not in a multi-issue one, a stronger partisan base always increases the benefit of taking the second issue off the table.  That does not necessarily mean a referendum enhances the electoral prospects of the majority party, however, since the benefit also depends on which party is relatively more divided on the second issue.

If the parties divide on the second issue, the fraction of Right voters who support their party's position of $y=1$ is $1-B(-\gamma-b_R)=B(\gamma+b_R)$ and the fraction of Left voters who support their party's position of $y=0$ is $B(-\gamma-b_L)$. So Right's base is more \emph{cohesive} than Left's if and only if:
\[
B(\gamma+b_R)\ge B(-\gamma-b_L)\iff \gamma\ge-\frac{b_R+b_L}{2}.
\]
When $b_R=-b_L$, Right is more cohesive than Left if $\gamma\ge 0$, and Left is more cohesive than Right if $\gamma\le 0$. Since $\gamma$ is symmetrically distributed around zero, each scenario is equally likely, and it is the majority party who benefits from a referendum. An increase in $b_R>0$ makes the Right party more cohesive for any $\gamma$ as more of its voters support $y=1$, and an increase in $b_L<0$ makes Left less cohesive as fewer of its supporters prefer $y=0$.  Thus, when $b_R>-b_L$ ($b_R<-b_L$), Right is expected to be more (respectively, less) cohesive in a multi-issue election than Left.

If $r=1/2$ and $b_R>-b_L$, Right is thus advantaged in a multi-issue election but not in a single-issue election, and so would prefer \emph{not} to take the second issue off the table.  As such, the cutoff for Right to initiate a referendum, $r_{\text{bind}}$, must exceed one half when $b_R>-b_L$.  Increases in $b_R$ and $b_L$ both increase Right's relative cohesion on the second issue, reducing the party's incentive to use a referendum to take that issue off the table (i.e., raising threshold $r_{\text{bind}}$). 


Finally, consider the comparative statics on partisan polarization.  When $b_R>-b_{L}$, so Right is the more united party, then $r_{\text{bind}}>1/2$ and they are indifferent over holding a referendum only when in the majority.  Part (2) of the proposition states that if inter-party polarization $p$ increases, then so does $r_{\text{bind}}$, meaning that an even larger majority is needed for Right to benefit from a referendum.  This is because polarization $p$ is irrelevant in a single-issue election: all policy voters always stick with their party.  In a multi-issue election, however, higher polarization $p$ increases the importance of the traditional issue, reducing cross-over voting and benefitting the larger party.  A referendum that takes the second issue off the table, turning a multi-issue election into a single issue one, thus gives a greater benefit to the small and divided party than to the large and united one.

\noindent\emph{To Summarize:} Binding referendums are only initiated when the parties initially divide on the second issue.  In this case they are are favored by large but divided parties, who have the most to gain by taking the second issue off the table.  

\section*{Non-Binding Referendums}

\begin{quote}
"\emph{In a 52-48 referendum this would be unfinished business by a long way. If the Remain campaign win two-thirds to one-third that ends it.}"
\end{quote}
\begin{flushright}
UK Independence Party (UKIP) Leader Nigel Farage, May 2016.\footnote{``Nigel Farage wants second referendum if Remain campaign scrapes narrow win'', \emph{Daily Mirror}, 16 May 2016.}
\end{flushright}
\medskip

\begin{quote}
"\emph{...the greatest democratic mandate ever seen.}"
\end{quote}
\begin{flushright}
 Farage on the \emph{Leave} campaign's 51.89-48.11 win, January 2020.\footnote{https://inews.co.uk/news/nigel-farage-brexit-parliament-square-speech-celebrations-392984}
\end{flushright}

We argued earlier that in most real-world contexts, policymakers are  imperfectly bound by direct democratic mandates.  Even if formally binding, policymakers can often find ways to circumvent them \citep{gerber2004direct}. How does a lack of commitment to honor direct democratic outcomes impact the electoral consequences of a referendum?  To address this question, we modify our benchmark model by assuming that each party implements its core supporters' majority-preferred policy both before \emph{and} after any referendum.\footnote{An alternative approach would be to assume that the party faces an exogenous cost of diverging from the winning policy, increasing in the difference between the size of the majority and one half, reflecting that larger public mandates are more costly to ignore. So long as this cost is not too large relative to differences in average preferences ($b_R-b_L$),  the parties will still diverge for some shifts in preferences ($\gamma$); all results extend, with the thresholds appropriately amended. This approach can span from the `binding' benchmark, in which costs of reneging are infinite, to `non-binding' in which these costs are zero.} 

If no referendum is held, the analysis is exactly as in the binding case, with Right's probability of winning given by \eqref{lambda} when $b_R<0$ and \eqref{eq:multi} when $b_R \geq 0$.  

What are the consequences of holding a non-binding referendum? Recall that the fraction of voters who favor $y=1$ is given by the left-hand side of expression \eqref{gammas}. This can be inverted to obtain $\gamma$'s realization for any vote share.  In other words: the referendum reveals the uncertain component of voters' preferences, captured by $\gamma$, shaping the subsequent electoral contest. Notice that in our framework non-binding referendums simply reveal voter preferences.  Thus, any harmful consequences of direct democracy can only emerge in our framework due to politics that transpire \emph{after} the referendum.  

\autoref{fig:0} shows the majority-preferred position among the electorate as a whole, as well as in each party, as a function of $\gamma$'s realization. The median policy voter in the electorate  prefers $y=1$ if and only if $\gamma$ exceeds the threshold $\gamma^*(r,b_R,b_L)$, defined in expression \eqref{gammas}. This threshold lies between $-b_R$ and $-b_L$, as highlighted in the figure. However, party $J$'s median core supporter prefers $y=1$ if and only if $\gamma\ge -b_J$. The referendum therefore yields one of three outcomes.

\begin{figure}[t!]
\centering
  \includegraphics[width=.9\linewidth]{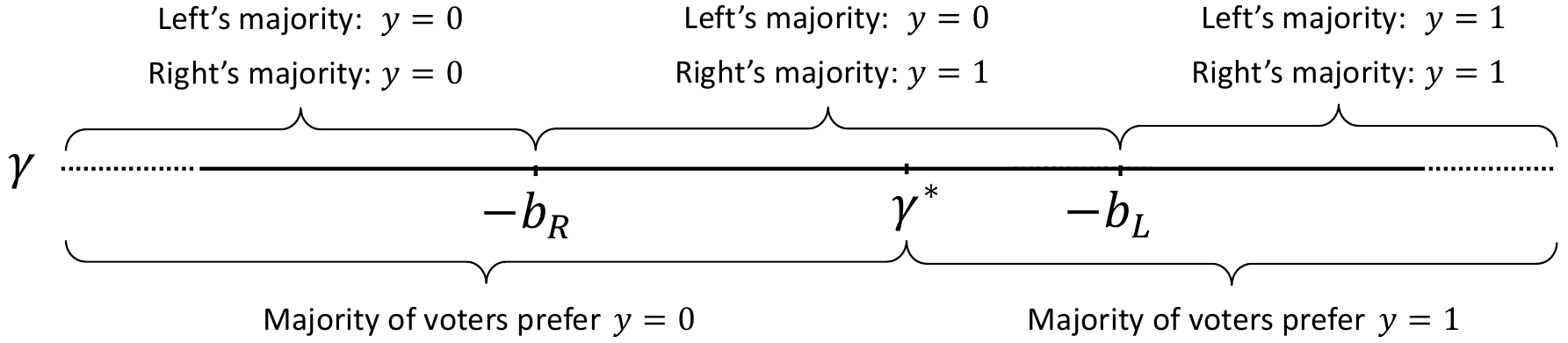}
\caption{How a non-binding referendum leads to either a single- or multi-issue election.}
\label{fig:0}
\end{figure}

(1) If $\gamma$ is sufficiently high, $\gamma\ge -b_L$, a referendum reveals that both parties' majorities favor $y=1$.  So the parties converge on $y=1$.

(2) If $\gamma$ is sufficiently low, $\gamma<-b_R$, a referendum reveals that both parties' majorities favor $y=0$.  So the parties converge on $y=0$.

(3) For intermediate $\gamma\in[-b_R,-b_L)$, $R$'s majority favors $y=1$ but $L$'s majority favors $y=0$.  So the parties diverge on both the partisan \emph{and} the second issue. 

Notice that in the first two cases, the referendum yields policy convergence on the second issue---just as if the parties were formally bound by its outcome. The ensuing election is fought solely on the traditional issue.  In the third case, the ensuing election is a multi-issue contest.  Since $\gamma^*\in(-b_R,-b_L)$, whenever a referendum outcome is sufficiently close to a fifty-fifty split, the parties always subsequently diverge.  Thus a close referendum ensures that the second issue will play a role in the ensuing election.  

\begin{remark}\label{r:close}
A sufficiently close referendum result \textbf{always} precipitates a multi-issue election.
\end{remark}

A single-issue election arises when the parties initially align on the second issue and no referendum is held, \emph{or} after a referendum that aligns the parties' majorities: $\gamma\ge -b_L$ or $\gamma<-b_R$. In either event, Right wins the election with probability $\lambda(r)$, defined in \eqref{lambda}.

A multi-issue election instead arises if (i) the parties initially mis-align on the second issue ($b_L<0\le b_R$) and no referendum takes place, or (ii) a referendum mis-aligns the parties' majorities by revealing $\gamma\in[-b_R,-b_L)$.  In these cases Right wins with probability $\lambda(s(\gamma;r,b_R,b_L,p))$, where $s(\cdot)$ is defined in \eqref{eq:noref}.  Since the realization of $\gamma$ is initially unknown, the expected winning probability involves integrating over possible $\gamma$ realizations.

How do these possibilities affect $R$'s electoral benefits from a referendum? The answer depends on whether the issue already divides the parties in the absence of a referendum.

\noindent\emph{Parties initially align on second issue.} Suppose, first, that $b_R<0$. A referendum has no impact on the election if it reveals $\gamma<-b_R$ or $\gamma\ge -b_L$, since the parties' policies on the second issue converge either way. When it reveals $\gamma\in[-b_R,-b_L)$, however, it divides the parties on the second issue and bundles it with the partisan one in the subsequent election. Right's net benefit from a referendum is 
\begin{align}
\bar{D}(r, b_L, b_R, p)&=\int_{-b_R}^{-b_L}(\lambda(s(\gamma;r,b_R,b_L,p))-\lambda(r))g(\gamma)\,d\gamma\propto \int_{-b_R}^{-b_L}(s(\gamma;r,b_R,b_L,p)-r)g(\gamma)\,d\gamma,  \label{eq:bn}
\end{align}
so Right will initiate a referendum if and only if $\bar{D}(r, b_L, b_R, p)>0$.

Because a binding referendum can only settle a conflict on the second issue, it has no electoral consequences when the parties initially align on that issue. By contrast, a non-binding referendum can only \emph{create} division on the second issue when the parties initially align. The trade-off between a single- and multi-issue election depends on the parties' relative size, and also their anticipated relative cohesion. The next result shows how this trade-off resolves.

\begin{proposition}\label{prop:bn}
When $b_L<b_R<0$, there exists a threshold $r^* \in \left(0, \frac{1}{2}\right)$ such that Right benefits from a referendum if and only if ${r}<r^*$. Threshold $r^*$ \emph{strictly decreases} in $b_R$, \emph{strictly increases} in $b_L$, and \emph{strictly decreases} in $p$. 
\end{proposition} 

\autoref{prop:bn} establishes that when parties initially align on the second issue, Right never initiates a referendum when in the majority. Surprisingly, it may not do so even when in the minority.

To see why, consider \autoref{fig:1}. The horizontal axis shows preference shock $\gamma$'s realization, the bell-shaped black curve is $\gamma$'s density, and the red line is Right's share of policy voters in a multi-issue election defined in \eqref{eq:noref}. This share strictly increases in $\gamma$, since Right's policy of $y=1$ aligns with a larger share of voters. The bold segments of the lines correspond to preference shock realizations $\gamma\in[-b_R,-b_L)$ that trigger a multi-issue election after a non-binding referendum: outside this interval, the referendum has no electoral consequence. The figure considers $r=1/2$,  so Right is more likely to win a multi-issue election than a single-issue election whenever the red line is above one half.

A referendum mis-aligns the parties on the second issue only when $\gamma\in[-b_R,-b_L)$. In that event, a fraction  $1-B(-b_R-\gamma)=B(b_R+\gamma)$ of Right's supporters prefer $y=1$ and a fraction $B(-b_L-\gamma)$ of Left's prefer $y=0$. This means that Right is more divided if $\gamma<-\frac{b_R+b_L}{2}$ and more united if $\gamma>-\frac{b_R+b_L}{2}$. If the $\gamma$'s density $g(\gamma)$ were constant, both cases would be equally likely.  But $g(\cdot)$ is strictly quasi-concave and peaked at zero, and thus strictly decreases on $[-b_R,-b_L]$, as highlighted by the decreasing bold black line in \autoref{fig:1}. As a result, $\gamma<-\frac{b_R+b_L}{2}$ is more likely than $\gamma>-\frac{b_R+b_L}{2}$. This means that Right expects to be less cohesive than Left---and thus relatively disadvantaged---in the event of a multi-issue election. This yields $r^*< 1/2$ and Right only initiates a referendum if strictly in the minority.

\begin{figure}[t!]
\centering
  \includegraphics[width=.7\linewidth]{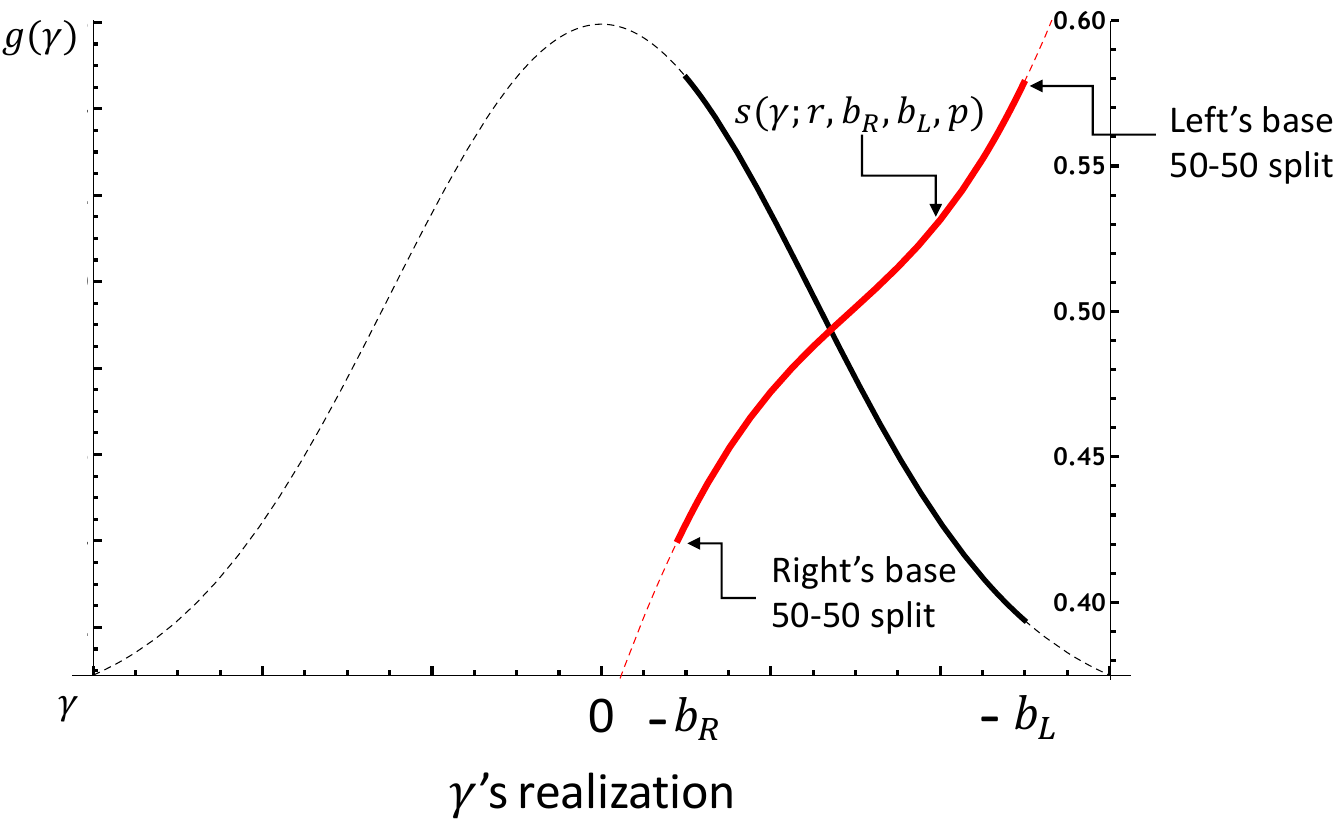}
\caption{Right's support after a referendum mis-aligns the parties' majorities on the emerging issue. The horizontal axis identifies $\gamma$'s realization; the left vertical axis is $g(\gamma)$ (\emph{black line}), and the right vertical axis is Right's share of policy voters after a referendum mis-aligns the parties' policies (\emph{red line}), $s(\gamma; r,b_R,b_L,p)$ defined in \eqref{eq:noref}. Here, $b_R=-.1$ and $b_L=-.5$. The thick lines correspond to realizations $\gamma\in[-b_R,-b_L]$, i.e., $\gamma\in [.1,.5]$ for which a referendum divides the parties on both issues. Primitives: $r=.5$, $p=.2$, for each $J\in\{L,R\}$, $b_i\sim N(b_i-b_J-\gamma,.2)$, $\gamma\sim N(0,.25)$.}
\label{fig:1}
\end{figure}


To understand the comparative statics in \autoref{prop:bn}, we consider 
how threshold $r^*$ changes as preferences shift inside either party. For concreteness, suppose $b_R$ increases to $\tilde{b}_R\in(b_R,0)$ and so a larger minority of Right's core supporters initially prefer $y=1$ (the effect of an increase in $b_L$ is the reverse). This has two effects, illustrated in \autoref{fig:2}. That figure again considers $r=1/2$, so that the parties win with probability one half in a single-issue election. 

The first effect of an increase from $b_R$ to $\tilde{b}_R$ is an improvement in Right's cohesion at all preference shock realizations in $[-{b}_R,-b_L)$. The reason is that, in a multi-issue election, higher $b_R$ means that a larger share of Right's base aligns with the party's majority in favor of $y=1$. Right therefore loses fewer of its core supporters in a multi-issue election. This improves Right's electoral prospects, reflected by an upward shift from the 
red line in \autoref{fig:2} to the blue line on the interval $[-{b}_R,-b_L]$. 

\begin{figure}[t!]
\centering
  \includegraphics[width=.7\linewidth]{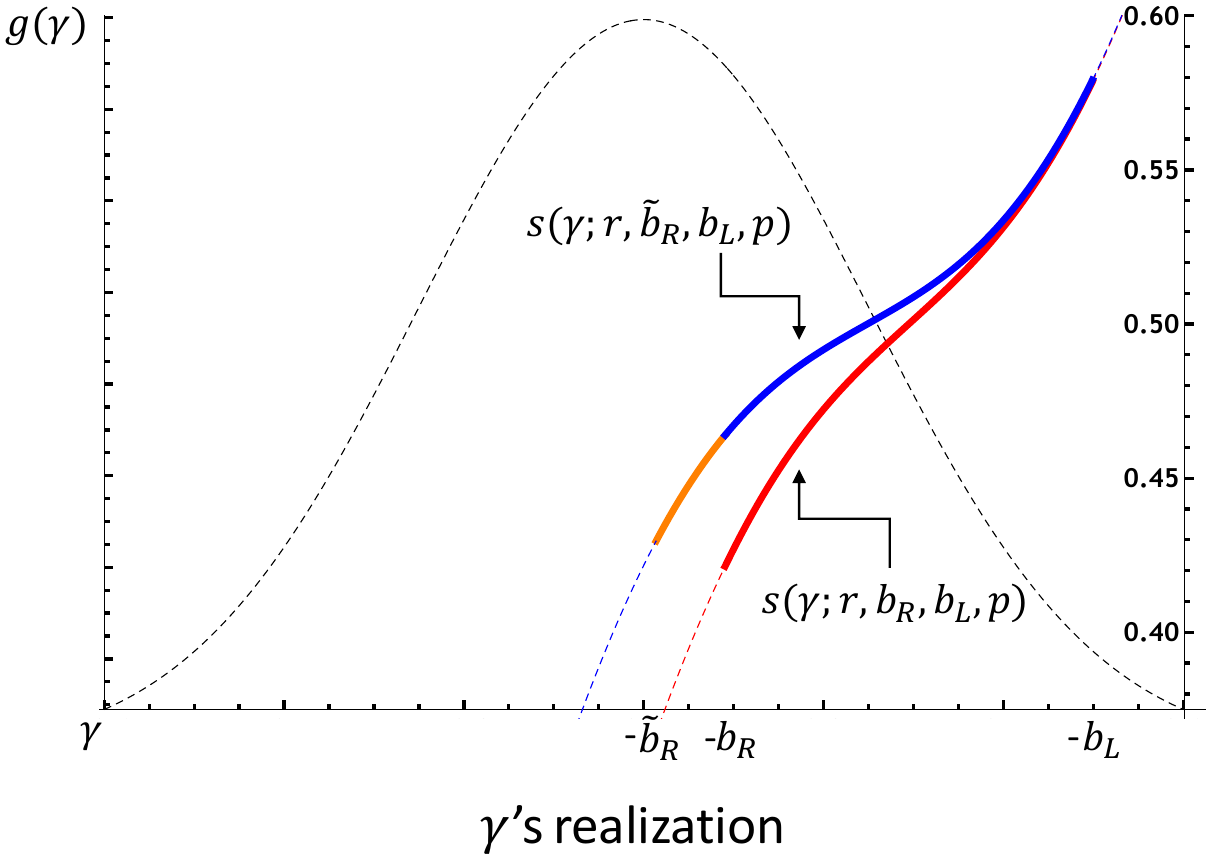}
\caption{How shifts in Right's cohesion change its incentive to hold a referendum. The horizontal axis identifies $\gamma$'s realization; the left vertical axis is $g(\gamma)$ (\emph{black line}), and the right vertical axis is Right's share of policy voters after a referendum mis-aligns the parties' policies, $s(\gamma; r,b_R,b_L,p)$ defined in \eqref{eq:noref}. The red line corresponds to $s(\gamma; r,-.1,b_L,p)$, and the blue/orange line corresponds to a shift to $s(\gamma; r,-.01,b_L,p)$. Primitives: $r=.5$, $p=.2$, for each $J\in\{L,R\}$, $b_i\sim N(b_i-b_J-\gamma,.2)$, $\gamma\sim N(0,.25)$.}
\label{fig:2}
\end{figure}

The second effect is to lower the threshold shock realization above which a referendum triggers a divided election from $\gamma=-b_R$ to $\gamma=-\tilde{b}_R$: if $\gamma \in [-\tilde{b}_R,-b_R]$ a referendum would not have triggered a multi-issue election, prior to $b_R$'s increase, and each party would win with probability one half. After the increase, however, these preference shocks lead to a multi-issue election. In \autoref{fig:2}, this is captured by the orange line $[-\tilde{b}_R,-b_R]$. Note that when $\gamma \in [-\tilde{b}_R,-b_R]$, Right is almost perfectly divided on the second dimension, and so disadvantaged by shifting from a single- to a multi-issue election. 

If density $g(\cdot)$ were constant the two effects would offset one another, but since $g(\gamma)$ strictly decreases on the interval the shocks that divide Right in a multi-issue election are relatively more likely.  Consequently the second (adverse) effect on Right's probability of winning dominates. Thus, Right's probability of winning the election after a referendum falls, and so $r^*$ \emph{strictly decreases} as $b_R$ increases.  

Finally, consider the effect of an increase in polarization.  Higher $p$ protects both parties' bases in a multi-issue election, making it more difficult for the minority party to steal its stronger opponent's core supporters. This reduces minority Right's benefit from a referendum, and so more intense inter-party polarization $p$ \emph{strictly lowers} threshold $r^*$.

In our office-motivated framework, Left's benefit from direct democracy is opposite to Right's. So, \autoref{prop:bn} implies that Left benefits from direct democracy if and only if $r>r^*$. Thus Left, the more cohesive party, may prefer to initiate a referendum even if in the majority, and is more likely to benefit from one when $b_R$ increases, so that Right's base  becomes more divided.

We turn to contexts in which the second issue already mis-aligns the parties' majorities.

\noindent\emph{Parties initially mis-align on second issue.} When $b_R\ge 0$ the parties initially face a multi-issue election. A referendum resolves the conflict if and only if it reveals a large enough shift in preferences to align the parties---either $\gamma<-b_R$ or $\gamma>-b_L$. Thus, the difference in Right's probability of winning the election after a referendum is 
\begin{align}
\tilde{D}(r, b_L, b_R, p)=\int_{-\infty}^{-b_R}(\lambda(r)-\lambda(s(\gamma;r,b_R,b_L,p))g(\gamma)\,d\gamma+\int_{-b_L}^{\infty}(\lambda(r)-\lambda(s(\gamma;r,b_R,b_L,p))g(\gamma)\,d\gamma,
\label{obj:BR}
\end{align}
and Right benefits from a referendum if and only if $\tilde{D}(r, b_L, b_R, p)>0$.

Recall that \autoref{pro:bench} shows that Right benefits from a binding referendum if and only if $r>r_{\text{bind}}$.  The next result states that a party prefers to use a non-binding referendum if and only if its partisan base is large enough. Paradoxically, a divided party can be even \emph{more} prone to use a non-binding referendum, even though it may fail to resolve the second issue.
 
\begin{proposition}\label{prop:bp}
If $b_R>0$, there exists $r^{**} \in (0, 1)$ such that Right benefits from a referendum if and only if ${r}>r^{**}$. Further, there exists $b_R^{\ddagger}<-b_L<b_R^{\dagger}$ such that $r^{**}<r_{\text{bind}}$ if $b_R^{\ddagger}<b_R<-b_L$, and  $r^{**}>r_{\text{bind}}$ if $-b_L<b_R<b_R^{\dagger}$.
\end{proposition}

The thresholds with a binding ($r_{\text{bind}}$) and non-binding ($r^{**}$) referendum are highlighted in \autoref{fig:3}. 
\autoref{prop:bp} shows that the divided party's incentive to hold a referendum can \emph{intensify} when it is non-binding and may fail to resolve the conflict.  


\begin{figure}[t!]
\centering
  \includegraphics[width=.8\linewidth]{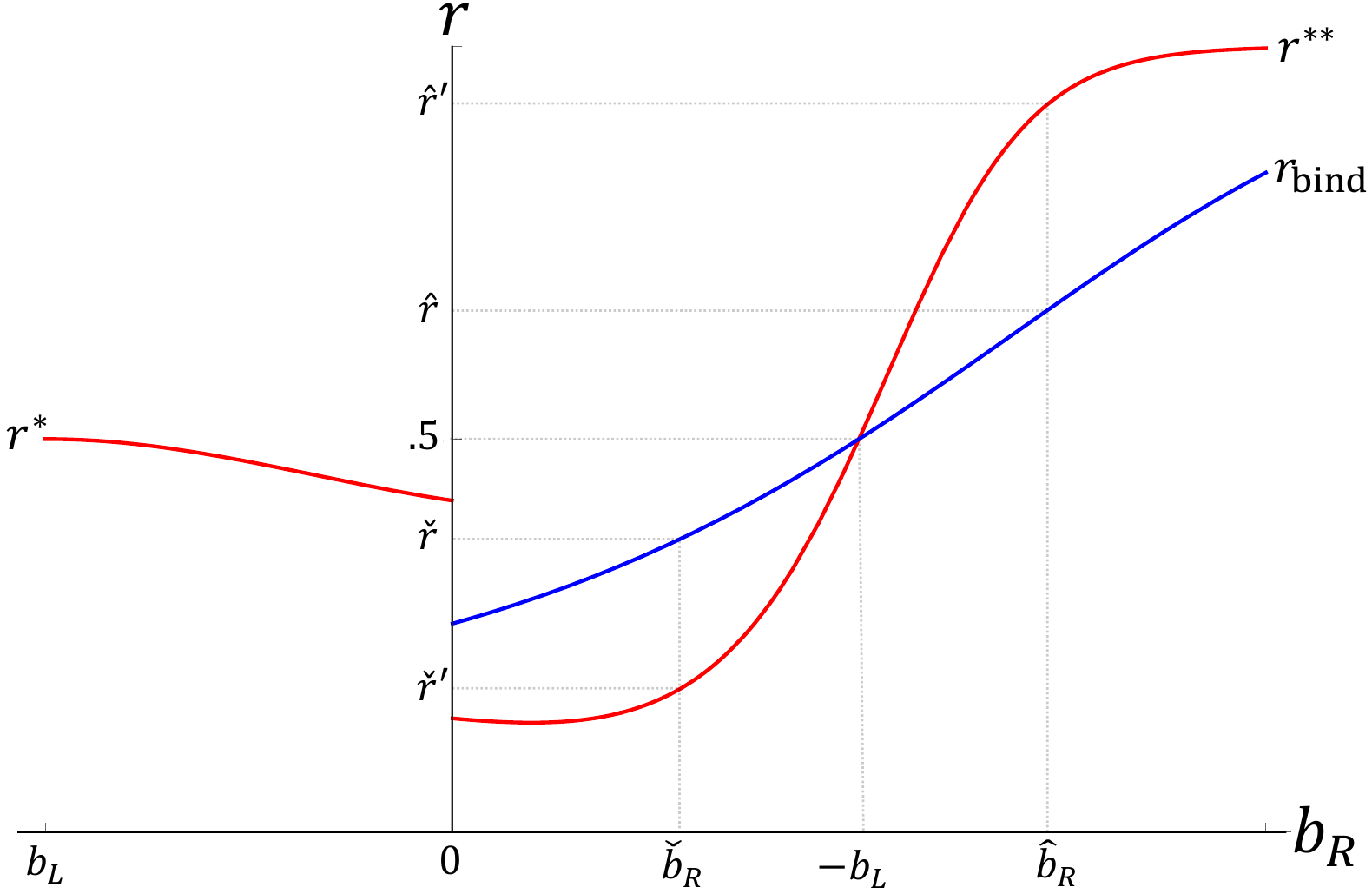}
\caption{Thresholds $r^*$ (\autoref{prop:bn}) and $r^{**}$ (\autoref{prop:bp}) are in red. The horizontal axis plots $b_R\ge -b_L$. Benchmark threshold $r_{\text{bind}}$ (\autoref{pro:bench}) is in blue. Primitives: $p=.05$, $b_L=-1$; for each $J\in\{L,R\}$, $b_i\sim N(b_i-b_J-\gamma,1)$, $\gamma\sim N(0,.5)$.}
\label{fig:3}
\end{figure}

To see why, suppose, $b_R=\check{b}_R<-b_L$, highlighted in \autoref{fig:3}, which implies that the party is indifferent between holding versus not holding a binding referendum when its share of core supporters is $\check{r}$. Because $\check{b}_R<-b_L$, threshold $\check{r}$ is strictly less than one half, reflecting that relatively divided Right is more vulnerable in a multi-issue election, and so prefers to take the second issue off the table unless severely disadvantaged in a single-issue election.

Binding referendums always resolve the second issue, but a non-binding referendum fails to resolve the second issue when \emph{both} parties are evenly split, i.e., when $\gamma\in(-b_R,-b_L)$. In those contexts, minority Right \emph{benefits} from a second salient electoral issue that  divides its stronger opponent. This makes a non-binding referendum relatively more valuable to minority Right than a binding one, and yields a lower threshold $\check{r}'<\check{r}$ than in the binding context. If, instead, $b_R=\hat{b}_R>-b_L$, so Right is relatively more united on the second issue, the logic reverses and $\hat{r}^{\prime}>\hat{r}$ on \autoref{fig:3}.\footnote{We are unable to establish a global result that $r^{**}>r^{\text{bind}}$ if and only if $b_R > -b_L$, although we believe it to be true. \autoref{prop:bp} instead establishes this ordering in a neighborhood of $b_R=-b_L$.}

\noindent \emph{Discussion.} California initiatives Proposition 187 and Proposition 8, discussed in the Introduction, illustrate how direct democracy can divide a strong candidate's base. Governor Pete Wilson trailed his opponent, Kathleen Brown, by 20 points a year before his November 1994 reelection.  Proposition 187 divided the parties in the November election, raising the GOP's appeal to marginal white voters \citep{tolbert1996race}.  
Wilson ultimately defeated Brown with a 14-point margin and Republicans captured control of the state legislature.\footnote{Because the referendum on Proposition 187 coincided with the election, an important dimension of the GOP's strategy was to mobilize turnout. We discuss this in a subsequent extension.}

The long-term prospects for Republicans in California were less favorable\footnote{Proposition 187 may have been partially responsible for these weakened prospects: see \url{shorturl.at/ipH16}} and by 2008, when Proposition 8 was put on the ballot by Republican activists, they were the clear minority party in California.  Consistent with \autoref{r:close}, Proposition 8's narrow passage ($52.2\%$ support) subsequently divided the parties in the 2010 Attorney General race, with Democratic candidate Kamala Harris pledging not to defend the voter-approved initiative.  While Harris ultimately prevailed, she did so with a margin of less than $1\%$ of the vote.  

Outside of the US, \cite{criddle1993french} documents EU referendums in France that  illustrate the logic of \autoref{prop:bn}. President Pompidou's 1972 EU enlargement referendum was partly intended to expose ``divisions among his political opponents (Socialists and Communists) in the run-up to parliamentary elections in which they were expected to make significant gains'' (p. 228). In 1992, President Mitterand faced the prospect of catastrophic defeat in upcoming parliamentary elections. A pre-election referendum on EU enlargement was partly an instrument ``to expose the policy divisions... between the parties of the right'' and reflected ``the continuation of domestic partisan objectives by other means; the issue had great mobilising [sic] potential for Mitterand, as it had a destabilising potential for the right'' (p. 229).  Years later, President Chirac's decision to convene a referendum on the EU Constitution in 2005 also seems to have succeed in dividing his main opponent, the Socialist party.  While the Socialist party officially backed the Constitution, a number of leading Socialists actively campaigned against it in the national referendum.  

Our framework also helps to interpret the aftermath of the 2016 Brexit referendum. Instead of settling the issue, Brexit's narrow victory in the referendum intensified its salience: Conservatives became the Brexit party, Labour the party of Remain. Contrary to expectations, the referendum revealed that a larger minority of Labour voters favored Brexit than the corresponding minority of Conservatives voters that favored Remain. In our framework, this amounts to a $\gamma$ realization smaller than $-b_L$, but closer to $-b_L$ than to $-b_R$ In the December 2019 election, 52 of the 54 previously Labour-held constituencies that elected a Conservative voted to Leave the EU in the 2016 referendum, and the Conservatives formed a majority government. The decision to hold the Brexit referendum was likely motivated, at least in part, by a third party threat (UKIP); we incorporate third party threats into the analysis in the extensions.

Finally, while the above examples highlight cases in which referendums fail to resolve the issue, as \cite{stokes2020} argues, Australia's voluntary "Marriage Law Postal Survey" in 2017 demonstrates that they can achieve a resolution.   In our framework this occurs when one alternative wins a clear majority, revealing information about $\gamma$ that causes the parties to converge.  While the survey was legally non-binding, after $61.6\%$ supported legalizing same-sex marriage, the Marriage Amendment Act codifying it passed with near unanimous support of the legislators in the two major parties--Labour and the Liberals.

\section*{Do Referendums Improve Congruence?} 

The  pre-eminent normative justification for direct democracy is that it better-aligns policy outcomes with the will of the majority \citep{besley2008issue, matsusaka2005direct}. This is because it can facilitate the majority-preferred policy on the issue subject to referendum, and by doing so, unbundle issues for voters improving alignment on other issues. We investigate this in the context of our framework. We say that a referendum \emph{improves congruence} on an issue if the (ex-ante) probability that the policy preferred by a majority of (policy-oriented) voters on that issue increases as a consequence of a referendum. 

Since a binding referendum ensures the majority preferred policy on the issue subject to referendum, it is immediate that it improves congruence on the second issue.

\begin{remark}\label{pro:refal}
A binding referendum strictly improves congruence on the emerging issue.
\end{remark}

What about the traditional issue?  When the parties initially mis-align on the emerging issue, a binding referendum ensures that the general election is fought only over the single issue.  One might conjecture that this always improves congruence on the traditional issue. This is false: aggregate uncertainty in single-issue elections due to noise voters means that the majority party is not guaranteed to win. As such, the majority-preferred policy may be more likely to win when the traditional issue is bundled with the second issue.  

A referendum improves congruence if and only if it improves the electoral prospects of the majority party. The previous characterization in \autoref{pro:bench} shows that a binding referendum may \emph{improve} or \emph{weaken} the majority party's prospects, and yields the following corollary.

\begin{corollary}\label{cor:trad_align1}
When $b_R<0$ a binding referendum has no impact on congruence on the traditional issue.  If $b_R>0$ a binding referendum improves congruence on the traditional issue when $b_R>-b_L$ if and only if $r \notin (1/2, r_{\text{bind}})$, and when $b_R<-b_L$ if and only if $r \notin ( r_{\text{bind}},1/2)$.
\end{corollary}

   Recall that Right benefits from a binding referendum if and only if $b_R>0$ and $r>r_{\text{bind}}$, illustrated in \autoref{fig:3}. Right is initially more divided when $0<b_R<-b_L$, and therefore prefers to take the emergent issue off the table even when it isn't the majority, reflected by threshold $r_{\text{bind}}<1/2$. So, when $r_{\text{bind}}<r<1/2$, a binding referendum raises minority Right's probability of winning the election, reducing congruence. Conversely, Right is initially more united when $b_R>-b_L$, and thus may prefer \emph{not} to take the emergent issue off the table even when it is a majority,  reflected by $r_{\text{bind}}>1/2$. So, when $1/2<r<r_{\text{bind}}$, a binding referendum reduces majority Right's probability of victory, reducing congruence. 

In addition to the above discussion, recall that non-binding referendum may introduce a new policy conflict into the election. Propositions \ref{prop:bn} and \ref{prop:bp} show when a non-binding referendum improves or weakens the majority party's chances, leading to the following corollary.


\begin{corollary}\label{cor:trad_align2}
When $b_R<0$ a non-binding referendum improves congruence on the traditional issue if and only if $r \in (r^*, 1/2)$.  When $b_R>0$ a non-binding referendum improves congruence when $b_R>-b_L$ if and only if $r \notin (1/2, r^{**})$ and when $b_R<-b_L$ if and only if $ r \notin (r^{**}, 1/2)$.
\end{corollary}

Recall that when $b_R<0$, Right favors a non-binding referendum if and only if its minority is sufficiently small $r<r^*$, illustrated in \autoref{fig:3}. A non-binding referendum therefore improves congruence only when it weakens minority Right's prospects, i.e., if and only if $r\in(r^*,1/2)$. The results for $b_R>0$ closely follow those for the binding case described in the previous Corollary, except with cutoff $r^{**}$ defined in \autoref{prop:bp}.

One might expect a referendum to at least improve congruence on the second issue.  This is true if the referendum is binding.  It is also true for a non-binding referendum when parties initially misalign on the second issue: a referendum at worst fails to settle the issue, and otherwise resolves policy in favor of the majority-preferred outcome.  However, if the parties initially align, a non-binding referendum may also weaken congruence on the issue subject to direct democracy. 

\begin{proposition}\label{ref:ladder}
When $b_R>0$, a non-binding referendum strictly improves congruence on the second issue. However if $b_R<0$ a non-binding referendum may improve or weaken congruence on the second issue.  \end{proposition}

Absent a referendum,  when $b_R<0$, the emergent policy outcome is $y=0$; recall that this is preferred by the majority if and only if $\gamma\le \gamma^*(r,b_R,b_L)$, defined in \eqref{gammas}.

To understand the effect of a non-binding referendum, note that if it reveals that $\gamma<-b_R$ then the policy outcome remains $y=0$, so there is no referendum-induced policy change.   If $\gamma>-b_L$, a non-binding referendum re-aligns the parties in favor of the majority-preferred policy $y=1$.  This occurs with probability $1-G(-b_L)$, yielding a congruent rather than non-congruent policy.

If $\gamma^*(r,b_R,b_L)<\gamma<-b_L$, the parties diverge and a majority of policy voters favor $y=1$. Absent a referendum, congruent policy $y=1$ is never implemented, but that policy is implemented after a non-binding referendum if and only if Right wins.  Right wins with probability $\lambda(s(\gamma;r,b_R,b_L,p))>0$, defined in \eqref{lambda} and \eqref{eq:noref}. 

However, if $-b_R<\gamma<\gamma^*(r,b_R,b_L)$, the parties diverge and a majority of policy voters favor $y=0$. Absent a referendum, the majority preferred policy $y=0$ is always implemented, but after a non-binding referendum $y=0$ is implemented if and only if Left wins, which occurs with probability $1-\lambda(s(\gamma;r,b_R,b_L,p))<1$. 

\begin{figure}[t!]
\centering
  \includegraphics[width=.8\linewidth]{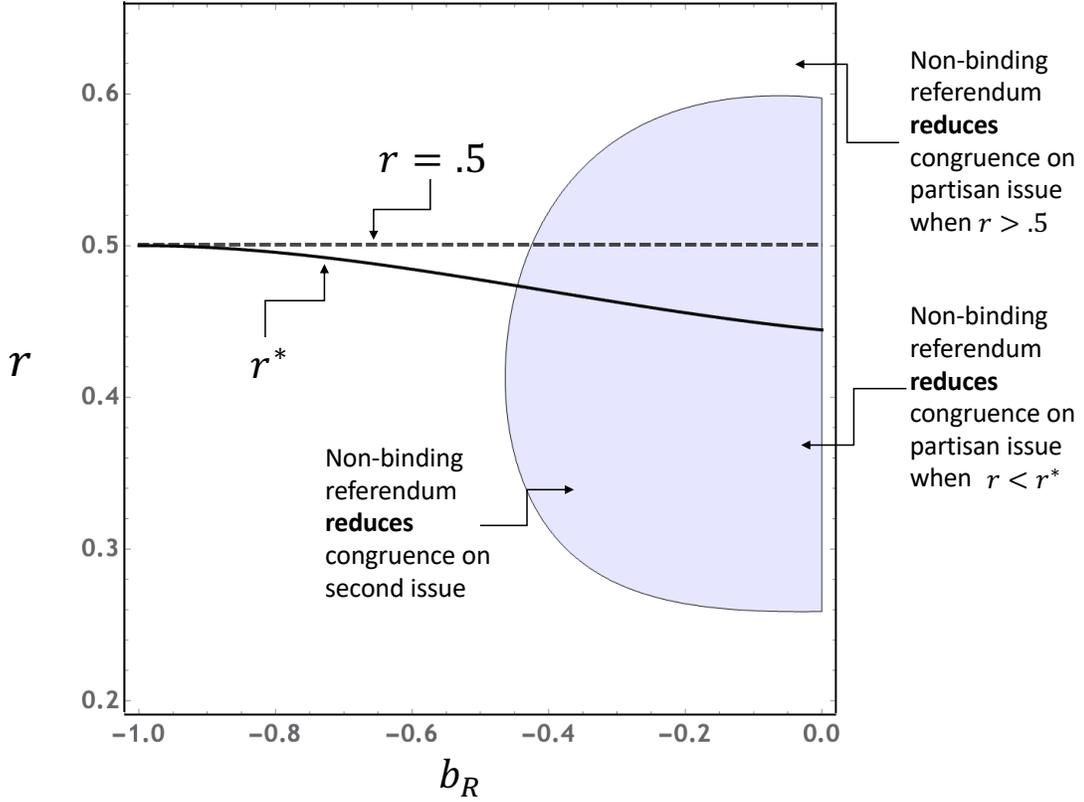}
\caption{\emph{Consequences of non-binding referendum for alignment on the emerging issue}. The figure plots threshold $r^*$  defined in \autoref{prop:bn}, such that Right benefits from a referendum if and only if $r<r^*$. The shaded region identifies  parameters such that \eqref{eq:Delta1} is strictly negative. Primitives: $p=1$, for each $J\in\{L,R\}$, $b_i-[b_K+\gamma]$ is standard Logistic, $\gamma\sim N(0,.5)$, $b_L=-1$, $\mu=.7$.} 
\label{fig:g}
\end{figure}

Putting this together, the net change in the probability of the majority-preferred outcome is
\begin{align}
1-G(-b_L)+& \int_{\gamma^*(r,b_R,b_L)}^{-b_L}\lambda(s(\gamma;r,b_R,b_L,p))g(\gamma)d\gamma-\int_{-b_R}^{\gamma^*(r,b_R,b_L)}\lambda(s(\gamma;r,b_R,b_L,p))g(\gamma)d\gamma
\label{eq:Delta1}.
\end{align}
Expression \eqref{eq:Delta1} can be negative if the third term is sufficiently large in magnitude.  That is, a non-binding referendum may lower congruence on the policy subject to the referendum! This happens when $b_L$ is very low, and $b_L$ and $b_R$ are quite different, and so the parties are likely to diverge on the second issue but $y=0$ is very likely to be majority-preferred.

\autoref{fig:g} illustrates that a referendum may either increase or decrease congruence on the second issue, i.e., that \eqref{eq:Delta1} may be positive or negative.
The red line plots threshold $r^*$  defined in \autoref{prop:bn}, such that Right benefits from a referendum if and only if $r<r^*$. The shaded region identifies parameters such that the probability of the majority-preferred outcome on the emerging issue is \emph{lower} after a non-binding referendum. 
Recall from \autoref{cor:trad_align2} that a non-binding referendum weakens congruence on the traditional partisan issue whenever $r \notin (r^*, 1/2)$.  Thus, for all parameters in the grey shaded region below $r^*$ or above $1/2$, congruence is simultaneously decreased on \emph{both} issues if a referendum is held.

\section*{Extensions} 

\noindent In a Supplemental Appendix, we pursue two extensions of our benchmark model.

\noindent\textbf{Third Party Threat.} In many cases, third parties form around emerging issue conflicts. These parties may be unlikely to win elections, but nonetheless disproportionately cannibalize the core supporters of one of the major parties. For example, UKIP mainly appealed to British Conservatives, while the Scottish Nationalist Party (SNP) tends to attract socially liberal (i.e., Labour-inclined) voters that support leaving the United Kingdom, but steals fewer votes from the Scottish Conservatives, whose members are largely pro-union.

To analyze these contexts, we augment our benchmark model by introducing an additional party, \emph{Third}, with fixed policy platform $x=1$ and $y=1$. This means that Third and Right have the same traditional partisan policy. We assume that Third has negative valence and therefore only wins support from voters who strictly prefer it on policy grounds. 
Finally, we study the most interesting context in which ${b}_L<0$ and ${b}_R<0$, so the two `major' parties, Left and Right, initially align on $y=0$. Thus, at the outset, only Third offers the policy $y=1$.  For example, prior to the Brexit referendum, both the British Conservative and Labour parties had official policies of `Remain' ($y=0$) while UKIP favored `Leave' ($y=1$).  

A referendum might reveal information that causes one or both major parties to position at $y=1$, eliminating Third's appeal.  Given that Third's policy aligns with Right's on the traditional issue, and appeals to relatively more of Right's voters on the second, Right loses a larger share of its voters to Third absent a referendum.  Especially when Right is in the majority, it therefore has an incentive to use a referendum in order to eliminate the Third party.   This is in contrast to our two-party baseline model, in which Right would never initiate a referendum when in the majority and the two parties initially align on $y=0$.

However, majority Right may prefer not to use a referendum if its base is very divided on the second issue relative to Left's, even though its relative vulnerability to Third is especially acute. The reason is that, if the referendum mis-aligns the parties on the second issue, reclaiming its core supporters from Third may be a pyrrhic victory for Right if in doing so it cedes a large enough share of its base to Left. 

\noindent\textbf{Timing and Mobilizing Turnout.} In many countries, referendums are held \emph{simultaneously} with ordinary elections.  This is common in California, for example, the US state that makes the most extensive use of direct democracy. Simultaneous timing implies that parties cannot adjust their electoral platforms in response to referendum results. Nonetheless, these referendums can impact elections by affecting who participates. 

Mobilization via voter initiatives is a pervasive tactic in the United States, such as the California Republican Party's support of Proposition 187,  discussed earlier. Similarly, in 1992 the California Democratic Party spent \$1 million to support a voter registration and ``get-out-the-vote'' drive against Proposition 165, which would have increased Governor Pete Wilson's ability to cut spending. In fact, the Federal Election Commission subsequently sued the party, arguing that the campaign was aimed at boosting turnout for Democratic candidates running for federal office \citep[741]{hasen2000parties}.  In 2004, eleven states held votes on initiatives outlawing same-sex marriage, viewed as benefiting George W. Bush's re-election by increasing turnout among social conservatives \citep{campbell2008religion}.  Motivating turnout was also likely the explanation for nationalist incumbent president Chen Shui-bian's decision to announce a referendum on relations with China on the same day as his own re-election contest in 2004 \citep{rawnsley2005peaceful}.  

To get at these issues, we pursue an extension in which a referendum is held simultaneously with the election. We focus on a binding referendum and assume that if a referendum is not held a status quo is maintained on the second issue.  Voters have random voting costs, and show up to vote if and only if the stakes from the outcome of the election (including the referendum if one is held) are large enough. A referendum raises the stakes and thus increases turnout among policy voters, which all else equal benefits the majority party.  Right thus benefits from a referendum if the fraction of conservatives is high enough.  However, turnout increases more when the average intensity of preferences in party $J$ --- reflected by $|b_J|$--- is higher.  Thus, even if they are in the minority,  Right benefits from a referendum if it has a sufficient intensity advantage over Left.   
The reason is that the second issue's presence on the ballot galvanizes the party's core supporters to vote in the referendum, and thus also in the ordinary election. 
%



\section*{Conclusion}

We analyzed the strategic use of direct democracy by an office-seeking party, and considered how this depends on whether parties are bound to honor the outcome of the referendum. We asked: under what circumstances can a party use a referendum to improve its electoral prospects? How does this depend on the extent of initial disagreement both across and within parties, as well as the party's initial electoral (dis)advantage? 

A referendum can be used to settle divisive issues. But, if non-binding, it can also create new dimensions of conflict. While strong and divided parties have an incentive to use referendums to take issues off the table, weak and cohesive parties can use them to create new conflicts.  Moreover, non-binding referendums may weaken congruence not only on the referendum issue but even on those ostensibly beyond the reach of direct democracy. 

Our analysis abstracts from many important features of real-world referendums. In practice, policymakers may be able to fine-tune the wording in ways that confer more or less ex-post discretion, allowing parties to decide not only whether to initiate a referendum but whether to make it binding or not. Another dimension policymakers may be able to influence is the timing of a referendum---whether to hold it simultaneously with the general election, hold a special election, or, as in the 2022 Kansas abortion referendum, simultaneously with the primary.   A party may then optimize timing to maximize the electoral benefit, or to minimize the electoral cost of a referendum that is required for the party's policy goals.  Finally, an important feature of referendums is that their success is often tied to the popularity and valence of the politicians who initiate them. We view these as important avenues for future research.  

\bibliographystyle{apsr}
\bibliography{BVW}

\newpage

\appendix 
\renewcommand{\theequation}{\Alph{section}.\arabic{equation}}
\renewcommand{\theproposition}{\Alph{section}.\arabic{proposition}}
\renewcommand{\theremark}{\Alph{section}.\arabic{remark}}
\setcounter{remark}{0}
\setcounter{equation}{0}

\begin{titlepage}
\section*{Supplementary Appendix to \emph{Pandora's Ballot Box} (For Online Publication)}


\subsection*{Table of Contents}

\begin{itemize}
\item[1.] Appendix A: Proofs of Propositions (\emph{page 1})
\item[2.] Appendix B: Referendum and Third Party Threat (\emph{page 9})
\item[3.] Appendix C: Direct Democracy and Mobilization (\emph{page 15})
\end{itemize}
\end{titlepage}
\newpage

\setcounter{page}{1}

\section{Proofs of Propositions}  
 
\noindent\textbf{Proof of \autoref{pro:bench}.} Right (strictly) benefits from a referendum when $b_R\ge 0$ if and only if $D(r, b_L, b_R, p)$ defined in \eqref{D} is strictly positive.  Note that 
\begin{align}
D(r, b_L, b_R, p) & \propto  \int_{-\infty}^{\infty}(r-s(\gamma, r, b_L, b_R, p))g(\gamma)d\gamma\nonumber \\
&= \int_{-\infty}^{\infty} [r (1-B(p+\gamma+b_R)+B(\gamma+b_L-p))-B(\gamma+b_L-p)]g(\gamma)d\gamma.\label{eq:condn}
\end{align}
It follows from inspection that $D(r, b_L, b_R, p)$ is strictly increasing in $r$, $b_L$, and $b_R$, and strictly positive if and only if 
\[
r>\frac{\int_{-\infty}^{\infty}B(\gamma+b_L-p)g(\gamma)d\gamma}{\int_{-\infty}^{\infty}(1-B(p+\gamma+b_R)+B(\gamma+b_L-p))g(\gamma)\,d\gamma}\equiv r_{\text{bind}}\in(0,1).
\]
It is immediate that $r_{\text{bind}}$ strictly \emph{increases} in each of $b_R$ and $b_L$.  

Now consider the case where $b_R=-b_L$.  Using $B(z)=1-B(-z)$ for all $z\in\mathbb{R}$ and $g(\gamma)=g(-\gamma)$ for all $\gamma\in\mathbb{R}$, we have:
\begin{align*}
\int_{-\infty}^{\infty}B(\gamma+b_L-p)g(\gamma)d\gamma&=\int_{-\infty}^{\infty}(1-B(p-\gamma-b_L))g(\gamma)d\gamma\\
&=\int_{-\infty}^{\infty}(1-B(p-\gamma+b_R))g(\gamma)d\gamma\\
&=\int_{-\infty}^{\infty}(1-B(p+\gamma+b_R))g(\gamma)d\gamma.
\end{align*}
We conclude that $r_{\text{bind}}=\frac{1}{2}$ when $b_R=-b_L$. 

We now turn to part (2).  Straightforward algebra reveals that
\begin{align}
\frac{\partial r_{\text{bind}}}{\partial p} \propto r_{\text{bind}}-\frac{\int_{-\infty}^{\infty}b(p-\gamma-b_L)g(\gamma)\,d\gamma}{\int_{-\infty}^{\infty}b(p-\gamma-b_L)g(\gamma)\,d\gamma+\int_{-\infty}^{\infty}b(-p-\gamma-b_R)g(\gamma)\,d\gamma} \label{eq:theonewewant}
\end{align}
Now recall that $r_{\text{bind}}=\frac{1}{2}$ when $b_R=-b_L$ and that $b(z)=b(-z)$ for all $z\in\mathbb{R}$ and $g(\gamma)=-g(\gamma)$ for all $ \gamma \in \mathbb{R}$.  Thus when $b_R=-b_L$:
\[
\int_{-\infty}^{\infty}b(p-\gamma-b_L)g(\gamma)\,d\gamma=\int_{-\infty}^{\infty}b(-p+\gamma-b_R)g(\gamma)\,d\gamma=\int_{-\infty}^{\infty}b(-p-\gamma-b_R)g(\gamma)\,d\gamma,
\]
and so the RHS of \eqref{eq:theonewewant} is zero. So if we can show that the RHS of \eqref{eq:theonewewant} is strictly increasing in $b_R$ we can conclude that $r_{bind}$ is strictly decreasing in $p$ when $b_R<-b_L$ and strictly increasing when $b_R>-b_L$.

Finally the derivative of the RHS of \eqref{eq:theonewewant} with respect to $b_R$ is:
\begin{align}
\frac{d r_{\text{bind}}}{d b_R}+\frac{\int_{-\infty}^{\infty}b'(-p-\gamma-b_R)g(\gamma)\,d\gamma\int_{-\infty}^{\infty}b(p-\gamma-b_L)g(\gamma)\,d\gamma}{[\int_{-\infty}^{\infty}b(p-\gamma-b_L)g(\gamma)\,d\gamma+\int_{-\infty}^{\infty}b(-p-\gamma-b_R)g(\gamma)\,d\gamma]^2}\label{eq:datin}.
\end{align}
As we already established that the first term is strictly positive it is sufficient to show that 
\[
\int_{-\infty}^{\infty}b'(-p-\gamma-b_R)g(\gamma)\,d\gamma\ge 0.
\]
To verify this inequality, define $u \equiv -p-b_R-\gamma$ and notice :
\begin{align}
\int_{-\infty}^{\infty}b'(-p-\gamma-b_R)g(\gamma)\,d\gamma=&\int_{0}^{\infty}b'(u)g(-p-b_R-u)\,du+\int_{-\infty}^{0}b'(u)g(-p-b_R-u)\,du\nonumber\\
=&\int_{0}^{\infty}b'(u)g(-p-b_R-u)\,du - \int_{0}^{\infty}b'(-u)g(-p-b_R+u)\,du\nonumber\\
=&\int_{0}^{\infty}b'(u)[g(-p-b_R-u)-g(u-p-b_R)]\,du\label{eq:final}
\end{align}
where the last equality uses that $b'(u)=-b'(-u)$ for all $u\in(-\infty,\infty)$. We have $b'(u)\le 0$ for all $u\ge 0$ and note that $-p-b_R-u<0$ for all $u\ge 0$, since $p>0$ and $b_R>0$. Since $g(x)$ increases in $x\le 0$ and decreases in $x\le 0$ and $g(x)=g(-x)$ for all $x\in\mathbb{R}$, we observe that $g(-p-b_R-u) \leq g(u-p-b_R)$ for all $u> 0$. We conclude that \eqref{eq:final} is weakly positive, and so $r_{\text{bind}}$ strictly increases in $p$ if $b_R > -b_L$; if $b_R < -b_L$, it strictly decreases in $p$. 
$\square$

\noindent\textbf{Proof of \autoref{prop:bn}}. The net benefit to Right of a referendum, $\bar{D}(r, b_L, b_R, p),$ defined in \eqref{eq:bn} is proportional to 
\begin{equation}
\int_{-b_R}^{-b_L}((1-r) B(\gamma+b_L-p)-rB(-p-\gamma-b_R))g(\gamma)\,d\gamma\label{eq:condnob}.
\end{equation}
\eqref{eq:condnob} is strictly decreasing  in $r$, strictly positive when $r=0$, and weakly negative when $r=\frac{1}{2}$ if and only if
\begin{align}
 \int_{-b_R}^{-b_L}(B(-p+\gamma+b_L)-B(-p-\gamma-b_R))g(\gamma)\,d\gamma\le 0\label{inequality0}.
\end{align}
Since $g(\gamma)$ weakly decreases for $\gamma\ge 0$ and $B(-p+\gamma+b_L)-B(-p-\gamma-b_R)$ strictly increases in $\gamma$, we may apply Chebyshev's integral inequality, which states that: 
\begin{align*}
 &\int_{-b_R}^{-b_L}(B(-p+\gamma+b_L)-B(-p-\gamma-b_R))g(\gamma)\,d\gamma\\
& \le \frac{1}{b_R-{b}_L}\int_{-b_R}^{-b_L}(B(-p+\gamma+b_L)-B(-p-\gamma-b_R))\,d\gamma\times \left(\int_{-b_R}^{-b_L}g(\gamma)\,d\gamma\right).
\end{align*}
Since $\int_{-b_R}^{-b_L}B(-p+\gamma+b_L)\,d\gamma-\int_{-b_R}^{-b_L}B(-p-\gamma-b_R)\,d\gamma=0$, it follows that \eqref{inequality0} holds.  From the intermediate value theorem we obtain that there exists a unique threshold $r^* \leq \frac{1}{2}$, below which \eqref{eq:condnob} is strictly positive and a referendum strictly increases Right's probability of winning. 

We now turn to comparative statics with respect to $b_R$ and $b_L$.  Re-arranging the condition when \eqref{eq:condnob} is equal to $0$, we get cutoff 
\[
r^*=\frac{1}{1+\frac{\int_{-b_R}^{-b_L}[B(-p-\gamma-b_R)g(\gamma)]\,d\gamma}{\int_{-b_R}^{-b_L}[B(-p+\gamma+b_L)g(\gamma)]\,d\gamma}}\equiv (1+\tau(b_L, b_R))^{-1}.
\]
As we have established that $r^* \leq 1/2$ it follows that $\tau(b_L, b_R) \geq 1$.  It is also immediate from inspection that 
\begin{equation}
\tau < \frac{B(-p)}{B(-p+b_L-b_R)}.  \label{in:tau}  
\end{equation}
And, applying L'Hopital's rule, 
\[
\lim_{b_R \rightarrow b_L} \tau(b_L, b_R)=\lim_{b_R \rightarrow b_L} \frac{B(-p)g(-b_R)-\int_{-b_R}^{-b_L}B'(-p-\gamma-b_R)g(\gamma)d\gamma}{B(-p-b_R+b_L)g(-b_R)}=\frac{B(-p)g(-b_L)}{B(-p)g(-b_L)}=1.
\]

Since $\tau(b_L, b_R) \geq 1$ it follows that $\tau$ is increasing in $b_R$ in a neighborhood of $b_L$.  We now show that $\tau(b_L, b_R)$ is increasing in $b_R$ for all $b_R \in (b_L, 0)$.  Given that $\tau$ is initially increasing, it is sufficient to show that $\tau$ is quasi-convex in $b_R$.  
   
Differentiating with respect to $b_R$: 
\begin{align*}
\frac{\partial \tau(b_L, b_R)}{\partial b_R} \propto &\int_{-b_R}^{-b_L}[B(-p+\gamma+b_L)g(\gamma)]\,d\gamma\times\left(B(-p)g(-b_R)-\int_{-b_R}^{-b_L}B'(-p-\gamma-b_R)g(\gamma)d\gamma\right)\\
&- \int_{-b_R}^{-b_L}B(-p-\gamma-b_R)g(\gamma)\,d\gamma\times\left(B(-p-b_R+b_L)g(-b_R)\right)\\
=\,&g(-b_R)B(-p)\int_{-b_R}^{-b_L}B(-p+\gamma+b_L)g(\gamma)d\gamma\\
&-g(-b_R)B(-p-b_R+b_L)\int_{-b_R}^{-b_L}B(-p-\gamma-b_R)g(\gamma)\,d\gamma \\
&-\int_{-b_R}^{-b_L}B(-p+\gamma+b_L)g(\gamma)\,d\gamma \int_{-b_R}^{-b_L}B'(-p-\gamma-b_R)g(\gamma)\,d\gamma\\
\propto\,&g(-b_R)(B(-p)-\tau(b_L, b_R)B(-p-b_R+b_L))-\int_{-b_R}^{-b_L}B'(-p-\gamma-b_R)g(\gamma)\,d\gamma \\
= & [g(-b_L)-g(-b_R)\tau(b_L, b_R)]B(-p-b_R+b_L)-\int_{-b_R}^{-b_L}B(-p-\gamma-b_R)g'(\gamma)\,d\gamma \\
\equiv &X(b_R).
\end{align*}
The last equality follows from integration by parts.

Since $\frac{\partial \tau}{\partial b_R}=0$ if and only if $X(b_R)=0$, to show that $\tau$ is quasi-convex in $b_R$ it is sufficient to show that $X'(b_R)>0$ when $\frac{\partial \tau}{\partial b_R}=0$.  Differentiating,  
\begin{align*}
X'(b_R)=& -[g(-b_L)-g(-b_R)\tau(b_L, b_R)]B'(-p-b_R+b_L)+g'(-b_R)\tau(b_L, b_R)B(-p-b_R+b_L) \\
&+ \int_{-b_R}^{-b_L}B'(-p-\gamma-b_R)g'(\gamma)\,d\gamma-B(-p)g'(-b_R) \\
=&g'(-b_R)[\tau(b_L, b_R)B(-p-b_R+b_L)-B(-p)] \\
&+\int_{-b_R}^{-b_L}B'(-p-\gamma-b_R)g'(\gamma)\,d\gamma-[g(-b_L)-g(-b_R)\tau(b_L, b_R)]B'(-p-b_R+b_L) \\
& > \int_{-b_R}^{-b_L}B'(-p-\gamma-b_R)g'(\gamma)\,d\gamma-[g(-b_L)-g(-b_R)\tau(b_L, b_R)]B'(-p-b_R+b_L).
\end{align*}
The last inequality follows from \eqref{in:tau} and that $g'(-b_R)< 0$ given that$ -b_R>0$.

Furthermore, when $X(b_R)=0$ it follows that 
\begin{align*}
[g(-b_L)-g(-b_R)\tau(b_L, b_R)]B'(-p-b_R+b_L)&=\int_{-b_R}^{-b_L}\frac{B'(-p-b_R+b_L)}{B(-p-b_R+b_L)}B(-p-\gamma-b_R)g'(\gamma)\,d\gamma \\
& \leq \int_{-b_R}^{-b_L}\frac{B'(-p-b_R-\gamma)}{B(-p-b_R-\gamma)}B(-p-\gamma-b_R)g'(\gamma)\,d\gamma \\
&=\int_{-b_R}^{-b_L}B'(-p-b_R-\gamma)g'(\gamma)\,d\gamma.
\end{align*}
The inequality follows because $g'(\gamma) \leq 0$, $-\gamma \geq b_L$ and $B(\cdot)$ is log-concave.  So
\[
X'(b_R) > 0.
\]
Thus we have that $\tau$ is initially increasing and is strictly quasi-convex in $b_R$, so we can conclude that $\tau(b_L, b_R)$ strictly increases in $b_R$.  Thus $r^*$ \emph{strictly decreases} in $b_R$. Note, in addition, that this implies that $\tau(b_L, b_R)>1$ for all $b_R>b_L$ and so we can conclude in fact that $r^*<1/2$.  

Now we turn to the comparative statics with respect to $b_L$.  Since $\tau(b_L, b_R) \geq 1$ it follows that $\tau$ is decreasing in $b_L$ in a neighborhood of $b_R$.  We now show that $\tau(b_L, b_R)$ is decreasing in $b_L$ for all $b_L<b_R$.  For this, it is sufficient to show that $\tau$ is quasi-convex in $b_L$. Similar to above, differentiating with respect to $b_L$: 
\begin{align*}
\frac{\partial \tau(b_L, b_R)}{\partial b_L} \propto &\int_{-b_R}^{-b_L}B(-p+\gamma+b_L)g(\gamma)\,d\gamma\times\left(B(-p-b_R+b_L)(-1)g(-b_L)\right)\\
&- \int_{-b_R}^{-b_L}B(-p-\gamma-b_R)g(\gamma)\,d\gamma\times \left[-g(-b_L)B(-p)+\int_{-b_R}^{-b_L}B'(-p+\gamma+b_L)g(\gamma)d\gamma\right]\\
 \propto & g(-b_L)(B(-p)-\tau(b_L, b_R)^{-1}B(-p-b_R+b_L))-\int_{-b_R}^{-b_L}B'(-p+\gamma+b_L)g(\gamma)\,d\gamma \\
= &-[g(-b_L) \tau(b_L, b_R)^{-1}+g(-b_R)]B(-p-b_R+b_L)+\int_{-b_R}^{-b_L}B(-p+\gamma+b_L)g'(\gamma)\,d\gamma \\
 \equiv & Y(b_L).
\end{align*}
As in the previous case, the last equality follows from integration by parts. 

Since $\frac{\partial \tau}{\partial b_L}=0$ if and only if $Y(b_L)=0$, to show that $\tau$ is quasi-convex in $b_L$ it is sufficient to show that $Y'(b_L)>0$ when $\frac{\partial \tau}{\partial b_L}=0$.  Differentiating, using that $Y=0$, and following similar steps as above, 
\begin{align*}
Y'(b_L) =&-[g(-b_L) \tau(b_L, b_R)^{-1}+g(-b_R)]B'(-p-b_R+b_L)+g'(-b_L)\tau(b_L, b_R)^{-1}B'(-p-b_R+b_L) \\
&-B(-p)g'(-b_L)+\int_{-b_R}^{-b_L}B'(-p+\gamma+b_L)g'(\gamma)\,d\gamma \\
=&g'(-b_L)[\tau(b_L, b_R)^{-1}B'(-p-b_R+b_L)-B(-p)] \\
&-(g(-b_L) \tau(b_L, b_R)^{-1}+g(-b_R))B'(-p-b_R+b_L)+\int_{-b_R}^{-b_L}B'(-p+\gamma+b_L)g'(\gamma)\,d\gamma \\
> &-[g(-b_L) \tau(b_L, b_R)^{-1}+g(-b_R)]B'(-p-b_R+b_L)+\int_{-b_R}^{-b_L}B'(-p+\gamma+b_L)g'(\gamma)\,d\gamma \\
=&-\frac{B'(-p-b_R+b_L)}{B(-p-b_R+b_L)}\int_{-b_R}^{-b_L}B(-p+\gamma+b_L)g'(\gamma)\,d\gamma+\int_{-b_R}^{-b_L}B'(-p+\gamma+b_L)g'(\gamma)\,d\gamma \\
\geq & 0.
\end{align*}
We conclude that $\tau(b_L, b_R)$ decreases in $b_L<b_R$, so $r^*$ \emph{increases} in $b_L$. 

Finally, we turn to comparative statics with respect to $p$. The derivative of \eqref{eq:condnob} with respect to $p$ is:
\begin{equation}
\int_{-b_R}^{-b_L}(rb(-p-\gamma-b_R)-(1-r)b(p-\gamma-b_L))g(\gamma)\,d\gamma\label{cp:eq},
\end{equation}
which strictly increases in $r$. Evaluating \eqref{cp:eq} at $r=\frac{1}{2}$, we find that 
\begin{align*}
\frac{1}{2}\int_{-b_R}^{-b_L} &[b(-p-\gamma-b_R)-b(p-\gamma-b_L)]g(\gamma)\,d\gamma \nonumber\\
\le &\frac{1}{2(b_R-b_L)}\left(\int_{-b_R}^{-b_L}[b(-p-\gamma-b_R)-b(p-\gamma-b_L)]d\gamma\right)=0.
\end{align*}
This follows from Chebyshev's integral inequality because $b(-p-\gamma-b_R)$ increases in $\gamma\in[-b_R,-b_L]$, $b(p-\gamma-b_L)$ decreases in $\gamma\in[-b_R,-b_L]$, and $g(\gamma)$ decreases in $\gamma\in[-b_R,-b_L]$. We conclude that $r^*$ strictly decreases in $p$. $\square$

\noindent\textbf{Proof of \autoref{prop:bp}.} When $b_R \geq 0$, by \eqref{obj:BR}, Right's net benefit from a non-binding referendum, $\tilde{D}(r, b_L, b_R, p)$, is proportional to
\begin{align}
\int_{\gamma \notin [-b_R, -b_L]}&(r B(-p-\gamma-b_R)-(1-r)B(-p+\gamma+b_L))g(\gamma)\,d\gamma.
\label{eq:thatcon}
\end{align}
\eqref{eq:thatcon} strictly increases in $r$, is strictly negative when $r=0$, and strictly positive when $r=1$. We conclude that there exists $r^{**}\in(0,1)$ such that \eqref{eq:thatcon} equals $0$ and that Right strictly benefits from a referendum if and only if $r>r^{**}$.

We now show that there exists $b_R^{\dagger}>-b_L$ such that if $b_R\in(-b_L,b_R^{\dagger})$, then $r^{**}>r_{\text{bind}}$, and a  $b_R^{\ddagger}<-b_L$ such that $b_R\in(b_R^{\ddagger},-b_L)$ implies $r^{**}<r_{\text{bind}}$.   To do so we evaluate $\tilde{D}(r, b_L, b_R, p)$ when $r=r_\text{{bind}}$: if it is strictly negative then $r^{**}>r_{\text{bind}}$ and if it is strictly positive then $r^{**}<r_{\text{bind}}$.

Evaluating \eqref{eq:thatcon} at $r_{\text{bind}}$, at which point \eqref{eq:condn} equals $0$, we get that $\tilde{D}(r_\text{{bind}}, b_L, b_R, p) \propto \Delta(b_R)$, where 
\[
\Delta(b_R)\equiv \int_{-b_R}^{-b_L}((1-r)B(-p+\gamma+b_L)-r B(-p-\gamma-b_R))g(\gamma)\,d\gamma.
\]
Notice that $\Delta(-b_L)=0$ when $r=1/2$, and so $r^{**}=r_{\text{bind}}=1/2$ when $b_R=-b_L$. Further:
\begin{align*}
\frac{\partial \Delta(b_R)}{\partial b_R}\bigg|_{b_R=-b_L,r=r_{\text{bind}}}\propto & -[B(-p)-B(-p+2b_L)] g(b_L)+\int_{b_L}^{-b_L}B'(-p-\gamma+b_L)g(\gamma)\,d\gamma \\
= & \int_{b_L}^{-b_L}B(-p-\gamma+b_L)g'(\gamma)\,d\gamma,
\end{align*}
where the last inequality follows from integration by parts.  
Furthermore, 
\begin{align*}
\int_{b_L}^{-b_L}B(-p-\gamma+b_L)g'(\gamma)\,d\gamma=&\int_{b_L}^{0}B(-p-\gamma+b_L)g'(\gamma)\,d\gamma+\int_{0}^{-b_L}B(-p-\gamma+b_L)g'(\gamma)\,d\gamma\\
< &\int_{b_L}^{0}B(-p-\gamma+b_L)g'(\gamma)\,d\gamma+B(-p+b_L)\int_{0}^{-b_L}g'(\gamma)\,d\gamma\\
=&\int_{b_L}^{0}[B(-p-\gamma+b_L)-B(-p+b_L)]g'(\gamma)\,d\gamma\\
<&0.  
\end{align*}

Therefore $\tilde{D}(r_{\text{bind}}, b_L, b_R, p)$ is negative for a neighborhood of $b_R>-b_L$ and positive for a neighborhood of $b_R<-b_L$.  We can thus conclude that there exists $b_R^{\dagger}>-b_L>b_R^{\ddagger}$ such that if $b_R\in(-b_L,b_R^{\dagger})$, then $r^{**}>r_{\text{bind}}$ and $b_R\in(b_R^{\ddagger},-b_L)$ implies $r^{**}<r_{\text{bind}}$.
$\;\square$

\noindent\textbf{Proof of \autoref{ref:ladder}.} 
That the effect of a referendum on congruence is ambiguous when $b_R<0$ is illustrated in \autoref{fig:g}.  

We prove that a referendum can only improve congruence when $b_R\ge 0$.  
Recall that $y=0$ is the majority preferred policy if and only if $\gamma<\gamma^{*}$, defined as the cutoff in which \eqref{gammas} holds with equality. If no referendum is held then Right wins an implements $y=1$ with probability $\lambda(s(\gamma;r,b_R,b_L,p))$, where $\lambda(\cdot)$ and $s(\cdot)$ are defined in expression \eqref{lambda} and \eqref{eq:noref} respectively.   
The probability that the majority-preferred emerging policy is implemented is:
\[
\int_{-\infty}^{\gamma^*}(1-\lambda(s(\gamma;r,b_R,b_L,p))g(\gamma)\,d\gamma+\int_{\gamma^*}^{\infty}\lambda(s(\gamma;r,b_R,b_L,p))g(\gamma)\,d\gamma.
\]
If a referendum is held both parties converge on the majority preferred policy if $\gamma<-b_R$ or $\gamma > -b_L$ so the corresponding probability is
\begin{align*}
\int_{-b_R}^{\gamma^*}(1-\lambda(s(\gamma;r,b_R,b_L,p))g(\gamma)\,d\gamma
+\int_{\gamma^*}^{-b_L}\lambda(s(\gamma;r,b_R,b_L,p))g(\gamma)\,d\gamma+1-G(-b_L)+G(-b_R), 
\end{align*}
and the net change in probability the majority-preferred policy on the second issue is implemented is
\[
\int_{-\infty}^{-b_R} \lambda(s(\gamma;r,b_R,b_L,p)g(\gamma)\,d\gamma+\int_{-b_L}^{\infty}(1-\lambda(s(\gamma;r,b_R,b_L,p))g(\gamma)\,d\gamma >0. \square
\]

\newpage 
\setcounter{equation}{0}

\section{Referendum and Third Party Threat}

Our benchmark model considers two-party electoral competition. Yet referendums are frequently deployed in contexts where one of the major parties' core supporters are especially vulnerable to being poached by a third party. For example, the UK's 2016 Brexit referendum was partly motivated by ongoing Conservative losses to the United Kingdom Independence Party (UKIP)---a right-leaning third party focused on leaving the European Union. 

We now extend our benchmark model to include another party, \emph{Third}. Third has a fixed policy of $x=1$ and $y=1$: this means that it shares Right's fixed policy on the traditional issue, but favors $y=1$. Third also has deterministic valence $v$, reflecting the utility difference to all voters from electing Third rather than one of the traditional parties.  We continue to focus on the non-binding context in which the two ``major'' parties, Left and Right, implement their supporters' (expected) majority-preferred policy on each issue.  Third maintains policy $y=1$ no matter what is revealed in the referendum.  

We focus on the case in which $b_L<b_R<0$, so a majority of each major party's core supporters are expected to favor $y=0$. This means that Third is initially the only party with policy $y=1$. We continue to assume that a fraction $\eta$ of noise voters support Right in a general election, and the remaining fraction $1-\eta$ support Left. Finally, we assume $v<0$, so that Third has a valence disadvantage versus the major parties. The magnitude of $v$ is not important, but $v<0$ ensures that if Right and Third offer the same policy, voters strictly prefer Right to Third. We maintain all other parameter and distributional assumptions from our benchmark model.

\noindent\textbf{Analysis.} Absent a referendum, Left and Right divide solely on the traditional policy. Among the share $\mu$ of policy voters, a fraction $r\in(0,1)$ are conservatives who prefer Right's (and therefore also Third's) traditional policy, and the remaining $1-r$ are liberals who prefer Left's traditional policy. As the major parties align on the second issue no policy voter will support the other party, but they may support Third.  A conservative voter $i$ supports Right if $0\ge b^i+v$; otherwise, she supports Third. Similarly, a liberal voter $i$ supports Left if $p\ge b^i+v$; otherwise, she supports Third.

\noindent We conclude that Right's total votes are:
\[
\mu r B(-v-b_R-\gamma)+(1-\mu)\eta,
\]
Left's total votes are:
\[
\mu(1-r) B(p-v-b_L-\gamma)+(1-\mu)(1-\eta),
\]
and Third's total votes are:
\[
\mu[r(1-B(-v-b_R-\gamma))+(1-r)(1-B(p-v-b_L-\gamma))].
\]
It follows that Right defeats Third if and only if
\[
\eta \ge \frac{\mu}{1-\mu}\left[r(1-2 B(-v-\gamma-b_R))+(1-r)(1-B(p-v-b_L-\gamma))\right],
\]
and that Right defeats Left if and only if
\[
\eta\ge \frac{1}{2}+\frac{\mu}{2(1-\mu)}\left[(1-r)B(p-v-b_L-\gamma)-r B(-v-b_R-\gamma)\right].
\]
Hence the probability Right receives more votes than Left, for any $\gamma$ realization, is 
\begin{equation}
\label{e:LR}
\hat{\lambda}(\gamma)=\frac{1}{2}-\frac{\mu}{2(1-\mu)}\left[(1-r)B(p-v-b_L-\gamma)-r B(-v-b_R-\gamma)\right].
\end{equation}  

Notice that Right loses a larger share of its core supporters to Third than Left for two reasons: $b_R>b_L$ implies a larger share of Right's core supporters prefer Third's policy of $y=1$, and also Left's core supporters mis-align with Third's traditional policy (reflected in $p>0$).  

This does not necessarily mean Third lowers Right's prospects relative to Left, however.  If Right is a small enough minority, even though they lose a larger \emph{fraction} of their voters they may expect to lose fewer \emph{total} voters to Third.  From \eqref{e:LR}, combined with the election probability in a two party election given by \eqref{lambda}, Right is less likely to finish ahead of Left in a three party race if and only if
\begin{equation}
\label{worse}
(1-r) \int_{-\infty}^{\infty} B(-p+v+b_L+\gamma) g(\gamma) d\gamma < r\int_{-\infty}^{\infty} B(v+b_R+\gamma) g(\gamma) d\gamma.
\end{equation}
Inequality \eqref{worse} holds if $r$ is not too small, and so Right has enough voters to potentially lose, or if polarization ($p$) is sufficiently high that Third attracts few voters from Left.  In particular, it always holds when $r \geq 1/2$ so Right is in the majority.  If \eqref{worse} holds, Right's electoral prospects are clearly harmed by the presence of Third.  Whether Left benefits depends on how likely Third is to defeat them.  Conversely if \eqref{worse}  is violated and Third is sufficiently unlikely to win---a sufficient condition is that $\mu$ is small---then Right's electoral prospects are enhanced by Third.

What about the incentive to hold a referendum?  If a referendum is held and reveals $\gamma< -b_R$,  it has no effect on election probabilities, since the parties' policy commitments remain the same. 

If a referendum is held and reveals $\gamma\ge -b_L$, both major parties' platforms converge on the second issue to $y=1$. Third wins zero votes, and Right wins with probability $\lambda(r)\equiv \frac{1}{2}+\frac{\mu}{1-\mu}\left[r-\frac{1}{2}\right]$
 as in a single-issue election in the baseline model.

If a referendum is held and reveals $\gamma\in[-b_R,-b_L)$, Right's policy becomes $y=1$, while Left's remains $y=0$, so that the major parties divide on both the traditional and the second issue. Since Third's valence $v<0$ is negative, it wins no votes.  Voters then divide across the two major parties as in our benchmark model, and Right's winning probability is
\[
\lambda(s(\gamma;r,b_R,b_L,p))=\frac{1}{2}+\frac{\mu}{1-\mu}\left[r B(p+\gamma+b_R)+(1-r) B(\gamma+b_L-p)-\frac{1}{2}\right].
\]

The above discussion has important implications when \eqref{worse} holds.  If the parties align on $y=0$ absent a referendum, Right has a greater incentive to initiate a referendum due to the Third party threat.  When $r>1/2$ Right would never initiate a referendum when $b_R<0$ absent a third party challenge, but may do so in the presence of Third.  In particular, Right always benefits from a referendum if $b_L \approx b_R$ and so Left and Right are unlikely to diverge on the second issue. 

\begin{remark}
\label{brexit}
If $r \geq 1/2$ and $b_L<b_R<0$, Right always benefits from a Referendum whenever $b_R-b_L$ is sufficiently small.
\end{remark}

This highlights how the threat of Third can induce a majority Right party to initiate a Referendum.  If the Referendum reveals that $\gamma$ is large enough that Right chooses $y=1$, the threat of Third is eliminated.  Consistent with this, after Brexit was initiated by the majority Conservative party, UKIP was effectively eliminated and the Conservatives enjoyed considerable electoral success.

However it is not the case that Right always benefits from a referendum. While a general analysis is complicated, since a referendum can change the winning probability for all three parties, we can give a full characterization by focusing on the following parameters.  First we assume that $\mu<2/3$, so the share of noise voters is not too small.  Second, we parameterize a scale family of distributions of voter types within each party, $B_{\sigma}\equiv B(z/\sigma)$ for $z\in\mathbb{R}$, and focus on the case in which $\sigma$ is high, and so preferences on the second issue are dispersed. 

Together, these two assumptions ensure that Third cannot win the election: since noise voters only support the major parties, $\mu<2/3$ means that Third can only succeed in winning the election if they unite a majority of policy voters in both parties, which isn't possible when voter preferences on the second issue are very dispersed.  It is then only necessary to determine if Right's chances are enhanced against Left, which is characterized in the next proposition.

\begin{proposition}\label{Pro:third} Suppose $\mu<2/3$. There exist $b_L^*<0$ and $b_R^* \in (b_L^*, 0)$ and $\sigma^*$ such that when $\sigma>\sigma^*$, 

\noindent (1)  \emph{advantaged} Right ($r>1/2$) strictly prefers to hold a referendum if $b_L>b_L^*$ or $b_R<b_R^*$, but strictly prefers not to if $b_L<b_L^*$ and $b_R^*<b_R<0$.

\noindent (2) \emph{disadvantaged} Right ($r<1/2$) strictly prefers to hold a referendum if $b_L<b_L^*$ and $b_R^*<b_R<0$ but strictly prefers not to if $b_L>b_L^*$ or $b_R<b_R^*$.

\end{proposition}

\begin{proof} First, we show that if $\mu<2/3$, there exists $\sigma^*$ such that Third can never win the election when $\sigma>\sigma^*$.  This follows because Third's vote share as $\sigma \rightarrow \infty$ is
\begin{align*}
\lim_{\sigma\rightarrow\infty} \mu[r(1-B_{\sigma}(-v-b_R-\gamma))+(1-r)(1-B_{\sigma}(p-v-b_L-\gamma))] &=\frac{\mu}{2}<\frac{1}{3}.
\end{align*}

As such, Right's net benefit from a referendum is
\[
\Gamma(b_R,b_L,r) \equiv \int_{-b_R}^{-b_L}[\lambda(s(\gamma;r,b_R,b_L,p))-\hat{\lambda}(\gamma; r, b_L, b_R, p)] g(\gamma)d\gamma\\
+\int_{-b_L}^{\infty} [\lambda(r)-\hat{\lambda}(\gamma; r, b_L, b_R, p)]g(\gamma)d\gamma.
\]
$\Gamma(b_R,b_L,r)$ is thus proportional to 
\begin{align*}
& \int_{-b_R}^{-b_L}\left\{
\begin{array}{c}
2r B_{\sigma}(p+\gamma+b_R)+2(1-r) B_{\sigma}(\gamma+b_L-p)-1\\
+(1-r)B_{\sigma}(p-v-b_L-\gamma)-r B_{\sigma}(-v-b_R-\gamma) 
\end{array}
\right\} g(\gamma)d\gamma \\
&+\int^{\infty}_{-b_L}\left\{ 2r-1+(1-r)B_{\sigma}(p-v-b_L-\gamma)-r B_{\sigma}(-v-b_R-\gamma) \right\}g(\gamma)d\gamma.
\end{align*}
And so
\begin{align*}
\lim_{\sigma \rightarrow \infty} \Gamma(b_R,b_L,r) &\propto  \int_{-b_R}^{-b_L}\left\{\frac{1}{2}-r\right\}g(\gamma)d\gamma+\int^{\infty}_{-b_L}\left\{ r-\frac{1}{2} \right\}g(\gamma)d\gamma \\
&=\left(r-\frac{1}{2}\right)\left[1-2G(-b_L)+G(-b_R) \right].
\end{align*}
It then follows when $r>1/2$ that Right benefits from initiating a Referendum if and only if $1-2G(-b_L)+G(-b_R)>0$, and if $r<1/2$ Right benefits if and only if $1-2G(-b_L)+G(-b_R)<0$. The final step to complete the proof is to determine when $1-2G(-b_L)+G(-b_R)\equiv \phi(b_L,b_R)$ is positive.

Notice that $\phi(b_L,b_R)$  strictly increases in $b_L$ and strictly decreases in $b_R$.  Moreover, $\phi(b_L,b_L)>0$ and $\phi(b_L,0)=1-2G(-b_L)+1/2=2G(b_L)-1/2$, which is strictly positive if and only if $b_L>G^{-1}(1/4) \equiv b_L^*$.  Thus if $b_L >b_L^*$ then $\phi(b_L,b_R)>0$ for all $b_R \in (b_L, 0)$.  If $b_L<b_L^*$ then $\phi(b_L, b_R)>0$ if and only if $b_R< b_R^*$, where $b_R^* \in (b_L^*, 0)$ solves $\phi(b_L, b_R^*)=0$.  \end{proof}

To understand the result, note that in the absence of a referendum there is no exchange of core supporters between the major parties but both major parties cede some of their core supporters that favor $y=1$ to Third.   When $\sigma$ gets large, and so preferences in each party are very dispersed, each party loses roughly the same share of its voters to Third, and so \eqref{worse} reduces to $r>1/2$.  That is, it is the larger party who is harmed by Third's presence. 

What is the effect of a referendum?  If it reveals $\gamma>-b_L$ then the major parties converge at $y=1$, Third receives no votes, and each policy voter supports their own party.  Thus if $\gamma>-b_L$, the majority party benefits from a referendum.  Conversely, if it reveals $\gamma\in[-b_R,-b_L)$ the major parties divide on the second issue.  Third again receives no votes, but now rather than losing a fraction of their support to Third, parties lose a fraction of their core supporters to the \emph{other major party}.  As losing votes to their main competition is worse than to a non-competitive third party, a referendum that divides the parties is worse for the majority party.  Thus, a majority party benefits from a referendum if and only if $\gamma>-b_L$ is sufficiently likely relative to $\gamma\in[-b_R,-b_L)$.

If Right is significantly more divided than Left (corresponding to $b_L<b_L^*$ and $b_R^*<b_R<0$), the referendum is likely to mis-align the parties on the second issue. When conservatives are in the majority, Right benefits from winning its core supporters back from Third, but doing so is a pyrrhic victory if the party loses a large share of its supporters instead to Left because they prefer $y=0$. In that context, it is better to accept Third's threat rather than risk a multi-issue conflict between the two major parties.

\newpage
\setcounter{equation}{0}

 \section{Direct Democracy and Voter Mobilization}  

In this extension, we explore how politicians can use direct democracy to stimulate turnout in elections that occur simultaneously.  We amend our benchmark model as follows. If a referendum is held, it occurs simultaneously with the election. We model turnout by way of a standard costly voting heuristic: each policy voter $i$ has an idiosyncratic voting cost $c^i$, distributed uniformly on $[0,\bar{c}]$ and independent of preference type $(x^i, b^i)$. Each policy voter casts a ballot if and only if the utility difference between her most preferred policy (or policies if both issues are on the ballot)  exceeds this voting cost.   Absent a referendum this utility difference is $p$, but if a referendum is being held the total utility difference from the two issues is $p+|b^i|$.

Since voting costs have bounded support, we further assume that the distribution of preference types $B(\cdot)$ has support $[-\sigma,\sigma]$, and that the distribution of shock $\gamma$ has support $[-\kappa,\kappa]$.\footnote{While the baseline model assumes unbounded support, all results extend with bounded support as long as $\sigma$ and $\kappa$ are sufficiently large.} We assume  $\sigma$, $\kappa$, and $\bar{c}$ are large in order to avoid corner solutions and ensure some fraction of every voter preference type $(x^i,b^i)$ turns up to vote:

\noindent\textbf{Assumption C.1.} $\bar{c}>p+\kappa+\sigma$, $\kappa>\max\{|b_R|,|b_L|\}$, and $\sigma>p+\kappa$.

\noindent The interaction proceeds as follows.

\noindent 0. Nature draws $\gamma$ and $\eta$. Neither of these realizations are observed by any agent.

\noindent 1. \emph{Right}'s leadership chooses whether to hold a referendum on the second issue.

\noindent 2. The general election takes place. If Right chooses to hold a referendum on the second issue, a binding referendum occurs simultaneously with the general election.

\noindent 3. The majority-winning party implements policy on the traditional issue.

As in our benchmark, each party always implements its fixed policy on the partisan issue, i.e., Right implements $x=1$ and Left implements $x=0$. We assume that $y=0$ is a ``status-quo'' policy that is implemented in the absence of any referendum.  This means that parties do not choose policy on the second issue absent a referendum, they are both constrained to the status quo.\footnote{This is one difference with the baseline model, and so applies to cases in which a referendum is a legal requirement to change policy---e.g., a constitutional amendment requiring ratification by referendum. We could alternatively assume parties are free to choose either $y=0$ or $y=1$ absent a referendum. If $b_L>0$, or if $b_R<0$, the parties align absent a referendum and the extension's analysis is unchanged.  If $b_L<0 \leq b_R$, preferences on the second issue would be relevant for turnout regardless of a referendum. A referendum unbundles the two issues as discussed in the previous literature.  Unbundling also increases turnout, since a voter participates if and only if $p+|b^i|\ge c^i$ after a referendum, but turns out if and only if $|p+b^i|\ge c^i$ in the absence of a referendum.}  We focus on a binding referendum on the second issue, so the majority-winning party after the election always implements the policy that commands majority support in the referendum.  Since the referendum and general are simultaneous, candidates would not have a chance to reposition after a non-binding referendum. 

\noindent\textbf{Analysis.} Absent a referendum, a policy voter casts a ballot for her preferred party if and only if her net value from her most-preferred partisan policy $p$ exceeds $c^i$. Thus, the fraction of policy voters that turn out for Right is $r p/\bar{c}$, and the fraction that turns out for Left is $(1-r)p/\bar{c}$. Right's probability of winning is therefore
\[
\frac{1}{2}+\frac{\mu}{1-\mu}\frac{p}{\bar{c}}\left(r-\frac{1}{2}\right).
\]
If a Referendum is held, a policy voter with net value $b^i$ from $y=1$ votes if and only if $p+|b^i|\ge c^i$. If she turns out, she always votes for the party whose traditional policy she prefers, and she votes sincerely on the second issue.  

For each voter $i$ in party $J$ define 
\[
u \equiv b^{i}-b_J-\gamma
\]
which is then symmetrically distributed from $[-\sigma, \sigma]$.  
The share of policy voters in the Right party who turn out for any $\gamma$ realization is 
\[
\frac{1}{\bar{c}} \int_{-\sigma}^{\sigma} (p+|u+b_R+\gamma|)d B(u)
\]
and so the fraction of policy voters who support Right is 
\[
\frac{r}{\bar{c}} \int_{-\kappa}^{\kappa}\left[\int_{-\sigma}^{\sigma}  (p+|u+b_R+\gamma|)d B(u)  \right] g(\gamma) d\gamma.
\]
Similarly, Left's expected fraction of policy voters is 
\[
\frac{1-r}{\bar{c}} \int_{-\kappa}^{\kappa} \left[\int_{-\sigma}^{\sigma} (p+|u+b_L+\gamma|)d B(u) \right] g(\gamma) d\gamma .
\]
This yields Right's probability of winning after a Referendum as
\begin{align*}
\frac{1}{2}+\frac{\mu}{1-\mu}\left(\frac{p}{\bar{c}}\left(r-\frac{1}{2}\right)\right)&+\frac{\mu}{1-\mu}\frac{r}{\bar{c}}\int_{-\kappa}^{\kappa}\left[\int_{-\sigma}^{\sigma}  |u+b_R+\gamma|d B(u) \right]g(\gamma) d\gamma \\
-\frac{\mu}{1-\mu}\frac{1-r}{\bar{c}} & \int_{-\kappa}^{\kappa} \left[\int_{-\sigma}^{\sigma} |u+b_L+\gamma| d B(u)  \right] g(\gamma) d\gamma.
\end{align*}
Right therefore benefits from a Referendum if and only if
\[
r\int_{-\kappa}^{\kappa} \left[\int_{-\sigma}^{\sigma} |u+b_R+\gamma| d B(u)  \right]g(\gamma) d\gamma
>(1-r)\int_{-\kappa}^{\kappa} \left[ \int_{-\sigma}^{\sigma} |u+b_L+\gamma| d B(u)  \right] g(\gamma) d\gamma
\]
or equivalently if and only if $r$ exceeds 
\begin{equation}
r_T(b_L, b_R) \equiv \frac{\int_{-\kappa}^{\kappa}\left[\int_{-\sigma}^{\sigma} |u+b_L+\gamma | d B(u) \right]g(\gamma) d\gamma}{\int_{-\kappa}^{\kappa}\left[\int_{-\sigma}^{\sigma} |u+b_L+\gamma | d B(u) \right] g(\gamma) d\gamma+\int_{-\kappa}^{\kappa}\left[\int_{-\sigma}^{\sigma}|u+b_R+\gamma | d B(u)  \right] g(\gamma) d\gamma}. \label{rthresh}
\end{equation}
We get the following result.

\begin{proposition}
Right benefits from a referendum if and only if $r>r_T(b_L, b_R) \in (0, 1)$.  Furthermore $r_T(b_L,b_R)=r_T(b_L, -b_R)=r_T(-b_L,b_R)$, $r_T$ strictly decreases in $|b_R|$ and strictly increases in $|b_L|$, and $r_T=1/2$ when $|b_R|=|b_L|$.
\end{proposition}
\begin{proof}
From \eqref{rthresh} it is immediate that $r_T \in (0, 1)$ and it follows from inspection that $r_T=1/2$ when when $b_R=b_L$.  Moreover, from the symmetry of $B(\cdot)$ and $G(\cdot)$, it follows that $r_T(b_L,b_R)=r_T(b_L, -b_R)=r_T(-b_L,b_R)$.  

All that remains to prove are the comparative statics.  First note that $r_T$ increases in 
\[
\int_{-\kappa}^{\kappa}\left[\int_{-\sigma}^{\sigma} |u+b_L+\gamma| d B(u)  \right]g(\gamma) d\gamma
\]
and decreases in 
\[
\int_{-\kappa}^{\kappa}\left[\int_{-\sigma}^{\sigma} |u+b_R+\gamma| d B(u)  \right]g(\gamma) d\gamma.
\]
Note that
\begin{align*}
& \int_{-\kappa}^{\kappa}\left[\int_{-\sigma}^{\sigma} |u+b_R+\gamma| d B(u)  \right]g(\gamma) d\gamma \\
=&\int_{-\kappa}^{\kappa}\left[-\int_{-\sigma}^{-(b_J+\gamma)}(u+b_J+\gamma) \, d B(u)+\int_{-(b_J+\gamma)}^{\sigma}(u+b_J+\gamma) \, d B(u)\right]g(\gamma)\,d\gamma,
\end{align*}
and the derivative of this expression with respect to $b_J$ is 
\begin{equation}
\label{e:deriv}
\int_{-\kappa}^{\kappa}\left[B(b_J+\gamma)-B(-b_J-\gamma)\right]g(\gamma)\,d\gamma.
\end{equation}
As \eqref{e:deriv} is equal to zero when $b_J=0$ and strictly increases in $b_J$ it follows that  
\[
\int_{-\kappa}^{\kappa}\left[\int_{-\sigma}^{\sigma}|u+b_J+\gamma|\,d B(u)\right]\,g(\gamma)d\gamma
\]
is strictly increasing in $|b_J|$.  Thus $r_T$ strictly increases in $|b_L|$ and decreases in $|b_R|$.  
\end{proof}

The intuition is straightforward.  $|b_J|$ reflects the median preference intensity of the core supporters in party $J$ on the second issue, and the party with the higher $|b_J|$ has the higher mean intensity.  A referendum increases turnout among policy voters, something which benefits the majority party when the average preference intensities are equal.    When one party has a greater average intensity however a referendum increases turnout more for that party than their rival.  The greatest incentive for Right to hold a referendum is then when it's the majority party  and its voters have intense preferences on the second issue.  However, even if in the minority, Right may prefer to initiate a referendum if the average intensity of its voters is high enough.

\end{document}